%% file: main.tex
\documentclass[english]{article}
\usepackage[T1]{fontenc}
\usepackage[utf8]{inputenc}
\usepackage[a4paper]{geometry}
\usepackage{babel}

\usepackage[tbtags]{amsmath}
\usepackage{mathtools}
\usepackage{amssymb, amsthm}
\usepackage{stmaryrd}
\usepackage{dsfont}

\usepackage{csquotes}

\usepackage[dvipsnames]{xcolor}
\usepackage{soulutf8}
\newcommand{\com}[1]{\textcolor{red}{[\hl{#1}]}}
\renewcommand{\com}[1]{}

\usepackage{tikz-cd}
\usepackage[sortcites=true]{biblatex}
\DeclareSourcemap{
  \maps[datatype=bibtex]{
     \map{
        \step[fieldsource=doi,final]
        \step[fieldset=isbn,null]
        }
      }
}

\input{def.tex}

\DeclareMathOperator{\Gal}{Gal}
\DeclareMathOperator{\gal}{\mathfrak{gal}}

\DeclareMathOperator{\Zrot}{Z_\text{rot}}

\DeclareMathOperator{\Zint}{Z_\text{int}}
\newcommand{\Zkin}{\Zint}

\DeclareMathOperator{\zrot}{z_\text{rot}}
\DeclareMathOperator{\zint}{z_\text{int}}
\newcommand{\zkin}{\zint}

\DeclareMathOperator{\Erot}{E_\text{rot}}

\DeclareMathOperator{\Uint}{U_\text{int}}

\newcommand{\tomega}{\tilde{\omega}}

\DeclareMathOperator{\EV}{\mathbb{E}}
\DeclareMathOperator{\grad}{grad}

\DeclareMathOperator{\SymOp}{\mathfrak{S}}

\newcommand{\Ebo}{\EV_{\beta, \omega}}
\newcommand{\dottimes}{\cdot}

\newcommand{\hf}{\hat{f}}
\newcommand{\hi}{\hat{\iota}}

\newcommand{\tobo}{\xrightarrow[\beta\omega^2 \to \infty]{}}
\newcommand{\simbo}{\underset{\beta\omega^2 \to \infty}{\sim}}
\newcommand{\eqbo}{\underset{\beta\omega^2 \to \infty}{=}}

\renewcommand{\Lie}{\operatorname{Lie}}
\newcommand{\Gibbs}{\Gamma}
\newcommand{\Rpp}{\setR_{>0}}
\newcommand{\cBR}{\overline{B}(0,R)}

\newcommand{\wprod}[2]{\left\langle #1 \, \wedge \, #2 \right\rangle}
\newcommand{\tprod}[2]{\left\langle #1 \, \otimes \, #2 \right\rangle}
\newcommand{\dprod}[2]{\left\langle #1 \, \dottimes \, #2 \right\rangle}

\title{The Hessian geometry of the ideal gas in rotation}

\author{Jérémie Pierard de Maujouy\\\textit{Université Catholique de Louvain, Belgium}}
\date{\today}

\begin{document}

\maketitle

\begin{abstract}
	We study the Hessian geometry associated with an ideal gas in a spherical centrifuge. According to Souriau, a spherically confined ideal gas admits states of thermal and rotational equilibrium. These states, called Gibbs states, form an exponential family of probability measures with an action of the Euclidean rotation group. We investigate its Hessian (Fisher-Rao) geometry and show that in the high angular velocity limit, the geometry can be compared to the Hessian geometry of a spherical rigid body, which we show to be isometric to a hyperbolic space.
\end{abstract}

\section{Introduction}

Symplectic and Poisson $G$-manifolds are largely accepted as the proper geometrical framework for modelling mechanical systems. In statistical mechanics, a large collection of identical bodies is modelled by a probability distribution on the space of states of the bodies. The probability distributions corresponding to thermal equilibrium are described by a particularly low (finite) number of degrees of freedom. In~\cite{SSDEng}, Souriau describes the construction of a manifold of probability distributions with maximal entropy on a Hamiltonian $G$-manifold. This manifold is called the Gibbs set, and its elements, called Gibbs measures, are meant to model equilibrium statistical states. These probabilities are indexed by an infinitesimal symmetry $\xi \in \Glie$ which is called a \enquote{generalised temperature}: when $G$ is the group of time translations, $\xi$ can be identified%
\footnote{We are assuming a unit system in which the Boltzmann constant $k_B$ is equal to $1$, identifying temperature with energy.}
with the inverse temperature $\beta = \frac{1}{T}$.%

As manifolds of probability distributions, Gibbs sets of Hamiltonian $G$-manifolds have the structure of statistical manifolds. More precisely, they form exponential families (as defined in~\cite{AJLSInformationGeometry,AmaNagaInformationGeometry}). As such, they are naturally equipped with the Fisher-Rao metric and with two dual flat affine connections for which the metric can be expressed as a Hessian.

Moreover, the Gibbs set has a natural Poisson structure and a Hamiltonian action of the Lie group $G$. The Fisher-Rao metric and the flat connections are invariant under this group action. 
In summary, Gibbs sets have a very rich $G$-equivariant geometry. There are, however, very few known examples of Gibbs sets for which this geometry has been studied (see~\cite{NilpOrb}).

In this paper, we study the Gibbs set associated with the symplectic manifold modelling free point-like massive particles confined in a ball. The symmetry group $G$ is the direct product $\setR\times \SO_3$ of the group of time translations and the rotation group. This is a direct adaptation of the ideal gas in a centrifuge described by Souriau~\cite{SSDEng} (see also~\cite{MarleSymplecticPoissonSouriau}). The gas in a cylindrical centrifuge, considered by Souriau, is the statistical mechanical system corresponding to free point-like massive particles confined in a cylinder, with as symmetry group the direct product $\setR\times\SO_2$ of the group of time translations and a rotation group around the cylindre axis. In the spherical case, the non-commutativeness of the symmetry group allows for a non-trivial Poisson structure on the Gibbs set, which is part of the structure we want to identify.

The central interest of this paper is, however, the Fisher-Rao metric, which is constructed as a Hessian. We compute it and express it in various coordinate systems. Its expression is difficult but it is possible to identify its asymptotic behaviour in a very specific limit, which can be understood as high angular velocity. In fact, we will proceed by identifying the limit of the Gibbs measures themselves, thereby giving a \enquote{partial compactification} of the Gibbs set. Further calculations will allow us to prove that in this limit, the metric and its derivatives up to the Riemannian curvature tensor behave as though the rotating gas is a rigid body with spherical symmetry. We prove that the Hessian structure of the suitable model for the rigid body is isometric to the hyperbolic space of dimension $4$. This leads us to conclude that the Gibbs set for the rotating ideal gas is \enquote{asymptotically hyperbolic} in a specific sense. Several aspects of our study, such as the fibred structure of the Gibbs set or the asymptotic behaviour of Gibbs measures, are suggestive of structures to look for in the future study of the geometry of more general Gibbs sets.

The paper is organised as follows. We start by introducing the general definitions and construction of Hessian geometry we will be using in Section~\ref{secno:HessGeoThermo}. It is illustrated with the relevant example of Ruppeiner's Riemannian geometry of thermodynamic systems, and we introduce Souriau's Gibbs sets and define the Hessian geometry we will be interested in. We also describe the Poisson structure of the Gibbs sets.
Section~\ref{secno:rigidbody} deals with the Hessian geometry of a rigid body, first assuming spherical symmetry, then in generality.
In this case, the space of equilibrium states is not constructed as a Gibbs set; therefore, we take great care to explain the choice of the flat connection used to construct the Hessian metric. The Hessian metric is proved to be isometric to the hyperbolic space of dimension $4$, and the Hessian geometry turns out to be dual to a well known model of the hyperbolic space. We also identify the Poisson structure and discuss the action of the rotation group.
We finally study the ideal gas in rotation in Section~\ref{secno:PerfGas}. The spherical confinement turns out to be incompatible with the group action: we work with a Hamiltonian $(\Lie(\setR)\times \so_3)$-manifold and identify the Gibbs set. We proceed to compute the Hessian metric. In order to obtain the asymptotic behaviour, we need to compute the differentials of the metric up to order $2$. 
We prove in Section~\ref{secno:CurvExp} a general formula for the iterated covariant differentials of the metric to any order, which we then apply to the case at hand. We conclude the paper by giving the asymptotic for the Riemannian curvature tensor and confirming that it is comparable to that of the rigid body, with however a contribution of the positional (in opposition to kinetic) degrees of freedom to the heat capacity of the body.

\tableofcontents
\subsection{Notations and conventions}\label{secno:notations}

Let us introduce here a few notations and conventions we will be using throughout the article.
We will be using the notation $\Rpp$ for the set of strictly positive real numbers (so as to avoid the notation $\setR^*_+$ which would conflict with the notation for dual vector spaces).
We will be freely identifying densities on a $n$-manifold, namely sections of the line bundle $|\Lambda^n T^*M|$, with their associated measure on $M$ (see for example \cite{LeeSmoothManifolds, Nicolaescu}).

Let $E$ be a vector space.
We will make use of the following symmetrising operators:
\begin{definition}\label{defno:SymOp}
	We write 
	\[
		\SymOp : \bigoplus_i E^{\otimes i} \to \bigoplus_i E^{\otimes i}
	\]
	for the \emph{unnormalised} symmetrising operator. If we write by $\sigma \cdot a$ the action of a permutation $\sigma\in \mathcal{S}_i$ on an element $a\in E^{\otimes i}$ then $\SymOp$ can be expressed as follows:
	\[
		\SymOp|_{E^{\otimes i}} (a) = \sum_{\sigma\in \mathcal{S}_i} \sigma\cdot a
	\]
\end{definition}

Furthermore, the following notation for symmetric products will be convenient: given $\alpha\in \Sym^j E, \gamma\in \Sym^k E$, we write
\[
	\alpha \dottimes \beta = \frac{\SymOp}{j!k!}\alpha\otimes\gamma
\]
This defines an associative commutative bilinear product. Note the absence of normalisation: for example with $\alpha \in E,\, \alpha \dottimes \alpha = 2 \alpha\otimes \alpha$. This convention will be helpful to minimise the appearance of factorial factors.

Finally, given a manifold $M$ and two vector space-valued multilinear forms $\alpha \in T^*_m M^{\otimes j}\otimes E, \gamma\in T^*_m M^{\otimes k}\otimes E^*$, we will use the notation $\tprod{\alpha}{\gamma}$ to denote the composition of the tensor product $T^*_m M^{\otimes j} \otimes T^*_m M^{\otimes j} \to T^*_m M^{\otimes (j+k)}$ with the contraction $E^*\otimes E \to \setR$. This notation will be also used with two $E$-valued multilinear forms when $E$ has an inner product. We will likewise use similar notations $\dprod{\alpha}{\gamma}$ and $\wprod{\alpha}{\gamma}$ for symmetric and exterior products composed with contraction.

It will be convenient to avoid using indices as much as possible, but indices will be useful when considering contractions on high order tensors.
We will then use Einstein's summation convention and write typically
\[
	\tprod{\alpha}{\beta} = \alpha_i \otimes \beta^i
\]

Finally, the Lie algebra $\so_3$ will be equipped with the standard inner product, which can be constructed according to the linear isomorphisms $\so_3 \simeq \Lambda^2 \setR^3 \simeq \setR^3$, or alternatively as minus half the Killing form. For example, the exponentials of the normalised elements are rotations of angle $1$ (in radians). It is invariant under the adjoint action. We also equip the dual $\so_3^*$ with the inverse inner product.


\section{Hessian geometry and thermodynamic system}\label{secno:HessGeoThermo}

Hessian geometry is the differential geometry framework that adapts to general differential manifolds the construction of a metric as the Hessian of a strictly convex function.

\subsection{Hessian geometry}\label{secno:HessGeo}
A general reference on Hessian geometry is~\cite{HessianStructures}, from which we quote a few definitions and results.

\begin{definition}[Affine manifold]
	An \emph{affine manifold} is a differentiable manifold equipped with an affine connection $D$ which is flat (curvature-less) and torsion-less.
\end{definition}
An affine manifold admits around every point local coordinates in which the connection coefficients vanish. Two such local coordinate systems always differ by an affine transformation applied to an open subset of $\setR^n$~
\cite{TransformationGroups, AJLSInformationGeometry}.
We call such coordinates \enquote{flat}.

\begin{definition}[Hessian of a function]
	Let $(M,D)$ be an affine manifold and $f : M \to \setR$ a smooth function.
	The second covariant derivative of $f$
	\[
		D^2 f : X,Y\in TM \mapsto D_X D_Y f - D_{D_X Y} f
	\]
	defines a section of $\Sym^2(T^*M)$, which is called the \emph{Hessian} of $f$.
\end{definition}

\begin{definition}[Convexity]
	We call a function $f : M\to \setR$ on an affine manifold 
	\emph{locally strictly convex} when its Hessian its everywhere definite positive.
\end{definition}%
In particular, the Hessian of a locally strictly convex function defines a metric on $M$. Hessian geometry deals with the geometry of such metrics.

\paragraph{Hessian curvature}

Hessian geometry consists in the study of Hessian manifolds, which we now introduce:
\begin{definition}[Hessian metric]
	Let $(M,D)$ be an affine manifold.
	A metric $g$ on $M$ is called a \emph{Hessian metric} if it can \emph{locally} be expressed as the Hessian of a function. A local function $\phi$ such that $D^2\phi = g$ is called a (local) \emph{potential}.

	An affine manifold equipped with a Hessian metric is called a \emph{Hessian manifold}.
\end{definition}

Let $(M, D,g)$ be a Hessian manifold.
Hessian manifolds have a notion of curvature that is an alternative to the Riemannian curvature:
\begin{proposition}[Hessian sectional curvature]\label{propno:HessCurv}%
	The \emph{Hessian curvature tensor} of $(M,D,g)$ is defined as
	\begin{equation}\label{eqno:HessianCurvature}
		K_{abcd} := \frac12\lp
				(D^2 g)_{abcd} 
				- g^{ef}(D g)_{ace}(D g)_{bdf} 
			\rp
		\in \Gamma\lp T^* M^{\otimes 4} \rp
	\end{equation}
	The totally covariant Riemannian curvature tensor of $g$ takes the form
	\[
		R_{abcd} = \frac12 \lp K_{abcd} - K_{bacd} \rp
	\]
\end{proposition}
Note that while the value of $K$ at a point $x$ depends on $g|_x$, $Dg|_x$ and $D^2 g|_x$, the value of $R$ at $x$ only depends on $g|_x$ and $D g|_x$.

In a similar fashion, Hessian manifolds have an alternative notion of sectional curvature and there is a characterisation of manifolds with constant Hessian sectional curvature:
\begin{proposition}\label{propno:seccurv}
	Let $h\in \Gamma\lp \Sym^2 T^*M \rp$. The \emph{Hessian sectional curvature} evaluated on $h$ is
	\[
		\kappa(h) := \frac{K_{abcd}h^{ac}h^{bd}}{\lVert h \rVert^2}
	\]
	The Hessian sectional curvature is a constant $c\in \setR$ if and only if
	\[
		K_{abcd} = \frac c 2 \lp g_{ab}g_{cd} - g_{ad}g_{cb} \rp
	\]
	In this case, the Riemannian manifold $(M,g)$ is a space form of constant Riemannian sectional curvature $- \frac c 4$.
\end{proposition}

\paragraph{Legendre transformation}

Let $(x^i)$ be flat coordinates on an open subset $\U \subset M$ such that the image of $\U$ in $\setR^n$ is a convex open subset and let $f : \U \to \setR$ be a locally strictly convex function. Then the application $Q : x\in \U \mapsto \lp \der{f}{x^i} \rp \in (\setR^n)^*$ is a diffeomorphism onto an open subset of $(\setR^n)^*$.
Denoting $(\epsilon^i)$ the dual canonical basis of $(\setR^n)^*$, this allows defining an alternative coordinate system $y_j = \der{f}{x^j}$.
The Hessian of $f$ can then be expressed as follows, using Einstein's summation convention:
\[
	D^2 f = \d y_i\otimes \d x^i
\]
But $D^2 f$ can also be expressed as a Hessian in the $(y_j)$ coordinates, thanks to the following identity:
\[
	\d (y_i x^i - f) = x^i\d y_i + y_i \d x^i - \underbrace{\der{x^i}{f}}_{y_i} \d x^i = x^i\d y_i
\]
Therefore, writing $\tilde D$ the connection on $\U$ for which the $y_j$ coordinates are flat, we obtain
\[
	\tilde D^2 (y_i x^i - f) = \d x^i \otimes \d y_i = D^2 f
\]
It turns out that $D$ and $\tilde D$ are adjoint connections for the metric $D^2 f$~\cite{HessianStructures, AmaNagaInformationGeometry}. The operation of replacing (locally) $D$ (or $x^i$) and $f$ respectively with $\tilde D$ (or $y_j$) and $y_i x^i - f$ is called \emph{Legendre transformation}.

\subsection{The Hessian geometry of a thermodynamic system}

The state of a thermodynamic system at equilibrium is characterised by a finite family of real parameters, such as temperature and chemical potentials. It is possible to use different sets of parameters, depending on the process that is to be modelled.
In~\cite{RiemannianThermodynamics, RuppeinerRiemannian} Ruppeiner proposes to study the Riemannian geometry defined by the Hessian of the entropy in a specific set of coordinates: the so-called \enquote{standard densities}. The idea is to give a more robust version of the quadratic approximation of entropy around its maximum that is commonly used in statistical physics.

Through the Legendre transformation, the metric is alternatively described as the Hessian of the Massieu thermodynamic potential $\phi$ in the dual coordinates, which include the inverse temperature $\beta = \frac 1 T$.

The first law of thermodynamics is often expressed through \enquote{fundamental thermodynamic identities}, which are expressions in various coordinate systems of a contact relation~\cite{GibbsContact, ArnoldThermo}.
The fundamental thermodynamic identity in terms of the thermodynamic potential $\phi$ is usually written as follows
:
\begin{equation}\label{eqno:FunThermo}
	\d \phi = - E \d \beta - \beta \mu_i \d N^i
\end{equation}
with $E\in \setR $ the energy of the system, $\beta\in \Rpp$ the inverse temperature%
,
$i$ being an index for different chemical components of the system, $N^i$ the number of particles of the chemical component $i$, and $\mu_i$ being what is called the chemical potential of the chemical component. However, the systems we will be interested in will not present chemical reactions but mechanical evolution. The chemical term in the relation~\eqref{eqno:FunThermo} will be ignored, but an extra term $-\delta W$ will be added to account for the evolution of mechanical energy~\cite{SSDEng}. The Hessian metric then takes the form
\[
	D^2 \phi = - \d E \otimes \d \beta - D\delta W
\]
with $D$ the affine connection that is flat in the dual coordinates.

\subsection{The Hessian geometry in Lie group thermodynamics}\label{secno:HessGeoLieThermo}

In~\cite{SSDEng} (see also the review~\cite{MarleGibbs1}) Souriau describes a geometric construction of the equilibrium statistical states of a mechanical system. We briefly introduce its main features.

A mechanical system is modelled by a symplectic manifold $(M, \omega)$ and a Lie group of symmetries $G\subset \Diff(M)$, preserving $\omega$ and admitting a (weak) momentum map $J : M \to \Glie^*$~\cite{IntroMechSym}.
A large collection of identical mechanical systems, each modelled by $M$, is statistically described by a probability measure on $M$.
According to the second law of thermodynamics, the spontaneous evolution of the large system has to increase the entropy of the probability measure (relative entropy with respect to the Liouville measure) while preserving the conserved quantities corresponding to symmetries of the process, in according with Noether's theorem.

The equilibrium states are described by the \emph{Gibbs distributions}, which have maximal entropy under the constraint of fixed $\EV[J]\in \Glie^*$. They take the following form:
\[
	p_\xi (x) = e^{-\sprod{\xi}{J(x)}} \frac{\lambda}{Z(\xi)}
\]
with $\lambda$ the Liouville measure of $M$ and $\xi \in \Glie$ such that the integral
\begin{equation}\label{eqno:PartFunIntegral}
	Z(\xi) := \int_M e^{-\sprod{\xi}{J}} \lambda
\end{equation}
converges. The \emph{Gibbs set} $\Gibbs \subset \Glie$ is defined as the set of $\xi\in \Glie$ such that the integral~\eqref{eqno:PartFunIntegral} converges \emph{on a neighbourhood of $\xi$}:
\[
	\Gamma = \operatorname{int}\lp \left\{
		\xi\in \Glie \suchthat Z(\xi) < \infty
	\right\}\rp
\]
It satisfies the following property:
\begin{proposition}
	The Gibbs set $\Gibbs$ is a convex open subset of $\Glie$.
\end{proposition}
The function $z:=\ln (Z)$, defined on $\Gamma$, is then invariant%
\footnote{Up to the eventual contribution of a symplectic cocycle.}
under the adjoint action of $G$ as well as strictly convex. It can therefore be used to define a Hessian geometry on $\Gibbs$.

Since $\Gamma$ is an open subset of $\Glie$, there is a natural parallelism $T\Gamma \simeq \Glie\times \Gamma$. As a consequence, the differential of $z$ can be seen as a map $\Gamma \to \Glie^*$. Under this identification, there is a Legendre transformation $\xi \in \Gibbs \mapsto Q(\xi) := -\d z(\xi) \in \Glie^*$ which turns out to coincide with the expected momentum:
\begin{equation}\label{eqno:dzEV}
	Q(\xi) = \EV_\xi[J]
\end{equation}
The dual potential%
\footnote{We follow the usual thermodynamic conventions for the dual potential which differ by a sign from the usual mathematical convention as described in Section~\ref{secno:HessGeo}.}
\begin{equation}\label{eqno:sLegTrans}
	s= z + \sprod{\xi}{Q}
\end{equation}
is found to equal the entropy of $p_\xi$ relative to $\lambda$.
As a consequence, the Hessian metric $g:=D^2 z$ coincides with the Hessian of the negative entropy in the $\Glie^*$-valued coordinates $Q$. The coordinates $\Gibbs\hookrightarrow \Glie$ are called \emph{generalised temperature}, while the dual coordinates in $\Glie^*$ form the \emph{expected momentum}.

In both Souriau's theory and Ruppeiner's theory, a Riemannian metric is constructed as the Hessian of the (negative) entropy in a specific set of coordinates. While Ruppeiner uses \enquote{standard densities} of thermodynamic nature, in Souriau's statistical mechanics the coordinates are the expected momenta, which is more of a mechanical nature. The relation between the two Hessian metrics is more obscure when they are expressed using the dual potentials that we wrote $\phi$ and $z$. We want, however, to take inspiration from Souriau's systems for \enquote{thermomechanical} systems, such as a rigid body with heat capacity. We will define the conserved quantities at the macroscopic level rather than deriving them from a microscopic mechanical system, but we will use the generalised temperature, namely the dual coordinates to the conserved momenta, as coordinates to define the Hessian metric. 

We are interested in the Hessian geometry of the Gibbs set of rotational thermomechanical equilibria of an ideal gas, which we will compute in Section~\ref{secno:PerfGas}.
However, we will first discuss a simpler model that will help us find a characterisation of the Riemannian geometry at infinity of the ideal gas.
This system can be understood as a rigid body, constructed at the macroscopic scale like the systems Ruppeiner considers, but with a state that is parametrised by its temperature and a mechanical degree of freedom, its rotation vector.

%
%

\subsection{Poisson structure of the Gibbs set}
\label{secno:PoissonStructureGibbsSet}

The dual vector space $\Glie^*$ admits an action $\Ad^\sharp$ of $G$ such that $J : M \to \Glie^*$ is equivariant: it is an action by affine transformations such that the linear part coincides with the standard coadjoint action. The translation part is encoded into a so-called \enquote{symplectic cocycle}~\cite{SSD, SymGeoAnaMech}.
It turns out that the application $Q : \Gibbs \to \Glie^*$ is open, intertwines the adjoint action and the $\Ad^\sharp$ action of $G$ and is a diffeomorphism onto its image $\Gibbs^*$~\cite{SSDEng,MarleGibbs1}. In particular, $Q$ can be used to pull back the Poisson structure from $\Gibbs^*\subseteq \Glie^*$ to $\Gibbs$. The map $Q$ can then be interpreted as a (Poisson) momentum map, with the same symplectic cocycle as $J$. The Gibbs set thus has a natural $G$-Hamiltonian Poisson structure. Its symplectic leaves are the connected components of the $G$-orbits.

It is easy to give an expression for the Poisson bracket. For simplicity, we only discuss the strong Hamiltonian case, namely, when the cocycle vanishes.
Let $f_1, f_2$ be two smooth functions on $\Gibbs$. We consider their pushouts to $\Gibbs^*$:
\[
	Q_* f_i = f_i\circ Q^{-1}
\]
therefore
\[\begin{aligned}
	\d (Q_*f_i)|_\alpha
		&= \d f_i \circ d Q^{-1} |_\alpha
		= \d f_i \circ g^{-1}|_{Q^{-1}(\alpha)} \in T^*_\alpha \Gibbs^*\\
		&{\mkern-4mu}\underset{DQ}{\equiv} \grad_g f_i \in T_{Q^{-1}(\alpha)} \Gibbs
\end{aligned}
\]
with $\grad_g f_i$ the \emph{Hessian gradient} of $f_i$ and $DQ : T\Gamma \isom T^*\Gamma^*$.
Finally, let us write $[\argdot, \argdot]_{\Glie}$ for the \emph{tensorial} bilinear application $T\Gibbs\otimes T\Gibbs \to T\Gibbs$ corresponding to the Lie algebra bracket under the parallelism $T\Gibbs \simeq \Gibbs\times \Glie$ defined by the open embedding $\Gibbs\hookrightarrow \Glie$. Then the Poisson bracket is
\begin{equation}\label{eqno:Poisson}
	\{f_1, f_2\}
		=
	\sprod{Q}{[\grad_g f_1, \grad_g f_2]_{\Glie}}
\end{equation}
We see that the Poisson structure contains information about the embedding of $\Gibbs$ into the Lie algebra $\Glie$.

\section{Rigid body in inertial rotation}\label{secno:rigidbody}

As a warm-up for the ideal gas in rotation, we consider a rigid body in rotation.
The system is not meant to be an accurate physical model or a proper thermodynamic model but applies Souriau's approach using symmetry and conserved momenta to a model constructed macroscopically (and not constructed as a space of Gibbs distributions).

For a statistical mechanical system symmetric under the action of the Galilean group $\Gal$, the equilibrium corresponding to a generalised temperature $\xi\in \gal$ coincides with the standard thermal equilibrium in a moving frame that follows a $1$-parameter group~\cite{SSDEng}. Typically, it can be a uniformly accelerated frame or a frame in constant uniform rotation. We see that this perspective is relevant for a rigid body in inertial motion as well: its constant rotation can be seen as an absence of motion in the suitably rotating frame.

\subsection{Rigid body with spherical symmetry}
\label{secno:rigidbodyspherical}

We start with the particular case of a rigid body with spherical symmetry. It will turn out that the ideal gas behaves asymptotically as though it were a rigid body with rotational symmetry.

The conserved momenta we will be using to characterise the state of the body are the total energy $E\in\setR$ and the angular momentum $M\in \so_3^*$. Under the Noether correspondence, they respectively correspond to the invariance under time translations and under rotations of the Euclidean space around a fixed origin point.

The thermal state of a rigid body at thermal equilibrium is described by a temperature $T \in \Rpp$ and the mechanical state can be described by a configuration $\theta\in \SO_3$, an angular velocity $\omega\in \so_3$ and the position of the centre of mass. Considering an isolated body, which is therefore on an inertial trajectory, we can ignore the motion of the centre of mass by working in the centre of mass reference frame. Furthermore, the configuration $\theta$ can be omitted when describing the equilibrium states since the system is invariant under rotation. Therefore, the parameters of the system are reduced to $(T, \omega)\in \Rpp\times \so_3$; they are, however, not the coordinates dual to $(E,M)$ in which we want to compute the Hessian metric. An educated guess would be to use the inverse temperature $\beta = 1/T$ instead of $T$, but let us give a derivation of the suitable coordinates.

The rigid body is characterised by an inertia tensor $I$ which is a Euclidean inner product on $\so_3$ and a heat capacity $C\in \Rpp$, which is therefore assumed to be independent of the temperature. Since the body is assumed to be spherically symmetric, the inertia tensor is (positively) proportional to any fixed invariant inner product on $\so_3$: we write $I(\omega, \omega) = I \omega^2$.

\paragraph{Dual coordinates}

The total energy is a sum%
\footnote{Here, we are building the model on \enquote{physical grounds} rather than proceeding by a mathematically justified construction.}
of two terms, the internal (thermal) energy $\Uint$ and a rotational kinetic energy term $\Erot$:
\[
	E(T, \omega) = \Uint + \Erot
\]
with
\begin{align}
	\Uint(T) &= CT\\
	\Erot(\omega) &=\frac12 I \omega^2
\end{align}
The angular momentum is (cf.~\cite{ArnoldMechanics}):
\[
	M = I(\omega, \argdot) \in \so_3^*
\]
The fundamental thermodynamic identity in term of $E$ is comprised of the heat transfer term $T\d S$ and of a \enquote{torque} term describing the variation of rotational kinetic energy
\[
	 T \d S + \sprod{\omega}{\d M}
	 	= \d E
		= \d \Uint + \d \lp \frac12 I \omega^2 \rp
		= \d \Uint + I(\omega, \d \omega)
\]
which gives the expected relation
\begin{equation}\label{eqno:dUTdS}
	\d \Uint = T \d S	
\end{equation}
but also
\[
	\d S
		= \frac{1}{T} \lp \d E - I(\omega, \d \omega) \rp
		= \beta \lp \d E - \sprod{\omega}{\d M} \rp
		= \beta \d E - \beta \sprod{\omega}{\d M}
\]
We have therefore identified the dual coordinates to $(E, M)$: $\beta = \frac{1}{T}\in \Rpp$ and $r := -\beta\omega \in \so_3$.
Although this model is not a statistical mechanical model, the dual coordinates correspond to the element $(\beta, -r)=\beta(1, \omega) \in \Rpp\times \so_3$ such that the body is at rest and thermal equilibrium in the rotating frame of reference given by $t\mapsto \exp(t \omega)\in \SO_3$.

\paragraph{The Hessian metric}

As explained in Section~\ref{secno:HessGeo}, the Hessian metric takes the following form:
\begin{equation}
	g = -\d \beta\otimes \d E  - \tprod{\d r}{\d M} 
\end{equation}
%
We compute the metric in the $(\beta, \omega)$ coordinates:
\[\begin{aligned}
	\d \beta\otimes \d E  + \tprod{\d r}{\d M} 
		&= \d \beta \otimes \lp \d \Uint + I \omega\d \omega \rp 
			+ \tprod{\d r}{I \d \omega}\\
		&= \d\beta \otimes \d \frac C \beta 
			+ \tprod{(\omega \d \beta - \d (\beta\omega))}{I\d \omega}\\
		&= -C \frac{\d \beta \otimes \d \beta}{\beta^2}
		- I \beta \tprod{\d \omega}{\d \omega}\\
		&= -\beta\lp
			\frac{C}\beta \frac{\d\beta}\beta \otimes\frac{\d\beta}\beta + I \tprod{\d \omega}{\d \omega} \rp
\end{aligned}\]
with the contraction $\tprod{\d \omega}{\d \omega}$ using the invariant inner product on $\so_3$ (see Section~\ref{secno:notations}).

We start recognising a familiar metric.
Let us define $u=\frac{2}{\sqrt\beta}$ 
so that 
$\frac{\d u}{u} = -\frac12 \frac{\d\beta}{\beta}$ and we can write
\[\begin{aligned}
	-\d \beta\otimes \d E  - \tprod{\d r}{\d M} 
	&= \frac{4}{u^2}\lp \Uint \lp -2\frac{\d u}{u} \rp 
							\otimes \lp -2\frac{\d u}{u} \rp 
						+ I \tprod{\d \omega}{\d \omega}
					\rp\\
	&= \frac{4}{u^2} \lp C\frac{u^2}{4} \frac{4}{u^2} \d u \otimes \d u 
						+ I \tprod{\d \omega}{\d \omega}
						\rp\\
	&= \frac{4C}{\tilde u ^2} \lp \d \tilde u \otimes \d \tilde u  + \tprod{\d \tomega}{\d \tomega} \rp
\end{aligned}
\]
with $\tilde u = \sqrt C u = 2\sqrt{\Uint}$ and $\tomega = \sqrt{I}\omega \in \so_3$. We have therefore identified $(\Rpp\times \so_3,g)$ with the $4$-dimensional hyperbolic half-space of curvature $-\frac1{4C}$.
The following forms will be relevant for the comparison with the case of the ideal gas in rotation:
\begin{equation}\label{eqno:RigidBodydM}
	-\d \beta\otimes \d E  - \tprod{\d r}{\d M} 
	= \beta \lp C \d u\otimes \d u + \frac{\sprod{\d M}{\d M}}{I} \rp
	= \beta \lp C \d u\otimes \d u + I\tprod{\d \omega}{\d \omega} \rp
\end{equation}

\paragraph{Massieu potential and entropy}

We have computed the metric, but we have not explicitly computed the Massieu potential or the entropy, of which the metric is the Hessian.
The entropy is readily expressed, up to a constant,%
\footnote{This constant is the object of the third law of thermodynamics; a constant heat capacity is not a suitable model for discussing this matter.}
using Equation~\eqref{eqno:dUTdS}:
\[
	\d S = \beta \d \Uint
		= - C\frac{\d\beta}{\beta}
		= \d \lp C\ln\lp \frac1\beta \rp \rp
		= \d \lp C\ln \lp E-\frac12 I^{-1}M^2 \rp \rp
\]
The expression for the entropy is, up to a constant, $S = C\ln \lp E - \frac12I^{-1} M^2 \rp$. This is a negative potential for the metric $g$ under the connection $\tilde D$ that is flat in the coordinates $(E, M)$. The expression of the potential is, up to a linear transformation, the same as that of the space of constant sectional Hessian curvature $\frac{1}{C}$ given in~\cite{HessianStructures} (Proposition 3.8, p. 47). In particular, the hyperbolic geometry can be explained by the Hessian sectional curvature associated to $(S, \tilde D)$ being constant (see Proposition~\ref{propno:seccurv} of the present paper).

The Massieu potential can be derived by integrating its thermodynamic identity~\eqref{eqno:FunThermo}, but it is more straightforward to build the Legendre transform (Equation~\eqref{eqno:sLegTrans}) of the entropy, which we have identified:
\[\begin{aligned}
	\phi 	&= S - \beta E - \sprod r M \\
		&= S - \beta\lp \frac C\beta + \frac12 I \omega^2 \rp + \beta\sprod \omega {I \omega}\\
		&= S - C + \frac12 \beta I \omega^2
\end{aligned}
\]
The Massieu potential can be compared to the potential given in~\cite{HessianStructures}, Example 6.7.

Moreover, the potentials and the metric are manifestly invariant under the adjoint action of $\SO_3$. The momentum coordinates are $\SO_3$-equivariant and define a Poisson structure for which the action of $\SO_3$ is Hamiltonian. Let us compute this Poisson structure.

\paragraph{Poisson structure}

The $(\setR\times\SO_3)$-equivariant open embedding $\Rpp\times \so_3 \overset{Q=(E, M)}{\hookrightarrow} \setR^*\times \so_3^*$ induces on $\Rpp\times \so_3$ a Poisson structure.
First, notice that the Poisson structure is entirely supported on the $\so_3$ component since the Poisson structure on the dual vector space $\setR^*$ is trivial. Furthermore, the angular momentum map $M$ is independent of $\beta$, therefore the Poisson structure can be discussed by forgetting entirely the $\setR^*$ component.

Rather than describing the Poisson bracket, which would depend on expressing the Lie bracket in coordinates, we describe the symplectic leaves and their symplectic forms. We will rely on the following unicity property:
\begin{proposition}\label{propno:SympSph}
	Let $\omega_1, \omega_2$ be two $\SO_3$-invariant symplectic forms on $\setS^2$. Then
	\[
		\omega_2 = \frac{\int_{\setS^2} \omega_2}{\int_{\setS^2}\omega_1} \omega_1
	\]
\end{proposition}
We know that the symplectic leaves are exactly the $\SO_3$-orbits, namely the $2$-spheres, and that the Poisson structure is invariant under $\SO_3$. Consequently, the corresponding volume forms on the spheres are proportional to the standard volume forms of the $2$-spheres, with a coefficient that can vary between $2$-spheres and thus depends on $\omega^2$. 

Let us write $\Lambda_{\so_3^*}$ for the canonical Poisson structure on $\so_3^*$: it gives the sphere of radius $\rho$ a surface of $4\pi \rho$. The inner product defines an isometry $\so_3\isom \so_3^*$: this induces on $\so_3$ a Poisson structure which we write $\Lambda_{\so_3}$.

The sphere with square radius $\omega^2 = \rho^2$ is mapped under $M$ to the sphere with $M^2 = I^2 \omega^2 = I^2 \rho^2$. As a consequence, the symplectic structure that $M^*\Lambda_{\so_3^*}$ induces on the sphere $\{\omega^2 = \rho^2\}$ has surface $4\pi I\rho$ and 
\[
	M^* \Lambda_{\so_3^*} = \frac 1 I \Lambda_{\so_3}
\]
Note that the \enquote{symplectic radius} $I\lVert \omega \rVert$ coincides neither with the Euclidean radius $\lVert \omega \rVert$ nor with the Hessian (metric) radius $I^{\frac12} \lVert \omega\rVert$.

\paragraph{Hessian curvature}

For future reference, let us compute the Hessian curvature on $\Rpp\times \so_3$. We have recognised the dual Hessian structure as a constant Hessian curvature geometry, but it is not the case for the (primal) Hessian structure. Let us write $D$ for the flat connection in the coordinates $(\beta, r)$. Recalling that $u^2 = 4/\beta$, we compute
\begin{align*}
	D^2 u
		= D \d u
		= D \lp \frac{-u^3}{8}\d \beta \rp
		= -\frac38 u^2 \d u\otimes \d\beta = 3\frac{\d u\otimes \d u}{u}
\end{align*}
and since $r+ \beta\omega = 0$,
\[
	0
		= D (\d r + \beta \d\omega + \omega \d\beta)
		= 0 + \d\beta\otimes \d \omega  +\beta D^2 \omega + \d \omega \otimes \d \beta
\]
and we obtain, using the symmetrising convention from Section~\ref{secno:notations}, 
\[
	D^2 \omega = -\frac{1}{\beta} (\d \beta\otimes \d\omega + \d \omega \otimes \d \beta) = \frac{2}{u}\d u\dottimes \d\omega
\]
These formulas allow us to compute
\[\begin{aligned}
	D g &= \d \lp \frac{4}{u^2} \rp \otimes \lp C\d u\otimes \d u + I \tprod{\d\omega}{\d\omega} \rp 
		 + \frac{4}{u^2} D\lp C\d u\otimes \d u + I \tprod{\d\omega}{\d\omega} \rp\\
		&\begin{aligned}
		= -\frac{8}{u^3} &\d u\otimes \lp C\d u\otimes \d u + I \tprod{\d\omega}{\d\omega} \rp \\
		 &+ \frac{4}{u^2} \Big( 6C\frac{\d u\otimes \d u \otimes \d u}{u}
		 + \frac{2I}{u} \big[ \tprod{(\d u \dottimes \d \omega)}{\d\omega} + \d u \otimes \tprod{\d\omega}{\d\omega} + \tprod{\d\omega}{\d \omega}\otimes \d u \big] \Big)
		 \end{aligned}\\		
		&= \frac{8}{u^3} \big( 2 C\d u^{\otimes 3}
			+ I  \d u\dottimes\tprod{\d\omega}{\d\omega}
			\big)
\end{aligned}
\]
and
\[\begin{aligned}
	D^2 g
		&= \d \lp \frac{8}{u^3}\rp \otimes \big( 2C \d u^{\otimes 3}
				+ I  \d u\dottimes\tprod{\d\omega}{\d\omega}
				\big)
			+ \frac{8}{u^3} D \big( 2 C\d u^{\otimes 3}
				+ I  \d u\dottimes\tprod{\d\omega}{\d\omega}
				\big)\\
		&\begin{aligned}
		&= -\frac{24\d u}{u^4}\otimes \big( 2 C\d u^{\otimes 3}
				+ I  \d u\dottimes\tprod{\d\omega}{\d\omega}
				\big)
			+ \frac{8}{u^3} \Bigg( \frac{2C\cdot 3^3}{u} \d u^{\otimes 4}\\
				&+ I 
				\d u \otimes \lp \frac{3}{u} \d u \dottimes\tprod{\d \omega}{\d \omega}
							+ \frac{2}{u} 2 \d u\dottimes\tprod{\d\omega}{\d\omega}
							\rp
				+ I\frac2u\tprod{\d\omega}{(
					\d u \dottimes \d u\dottimes \d\omega
					)}
				 \Bigg)
		\end{aligned}
		\\					
		 &= 
	\frac{16}{u^4} \Big( 24 C\d u^{\otimes 4}
		+ 2 I \lp \d u \otimes \d u \rp \dottimes \tprod{\d\omega}{\d\omega}
			\Big)
\end{aligned}\]

The Hessian curvature $K$ is then obtained according to Formula~\eqref{eqno:HessianCurvature}.
Let us first rewrite $D g$:
\[
	D g = \frac{8}{u^3}\big(
		(2C\d u^{\otimes 2}
		+ I \tprod{\d \omega}{\d\omega})
		\otimes \d u
		+
		I \tprod{(\d u\dottimes \d \omega)}{\d\omega}	
	\big)
\]

Using indices, we want to obtain the contraction $g^{mm'}D_ig_{km} D_j g_{lm'}$. The expression of $g$ in terms of $\d u$ and $\d \omega$ was given in Equation~\eqref{eqno:RigidBodydM}. The calculations yield the following expressions:
\[\begin{aligned}
	\lp \frac{8}{u^3}\rp^2
		&\lVert \d u \rVert^2
		\Big( 4C^2 \d u^{\otimes 4}
		+ 2CI \lp \tprod{\d u\otimes \d \omega}{\d u \otimes \d \omega}
				+\tprod{\d \omega\otimes \d u}{\d\omega\otimes \d u}
			 \rp\\
		&\qquad\qquad+ I^2 \d \omega_i \otimes \d\omega_j \otimes \d\omega^i \otimes \d\omega^j 
		\Big)\\
		&+
		\lp \frac{8}{u^3}\rp^2\frac{\lVert \d \omega \rVert^2}{3}
		I^2 \Big( \d u \otimes \d u \otimes \tprod{\d \omega}{\d\omega}
			+ \d u \otimes \tprod{\d\omega}{\d\omega}\otimes \d u\\
			&\qquad\qquad+ \tprod{\d \omega \otimes \d u}{\d u \otimes \d \omega}
			+ \tprod{\d\omega}{\d\omega}\otimes \d u\otimes \d u
		\Big)\\
	&= \frac{16}{Cu^4}
		\Big( 4C^2 \d u^{\otimes 4}
		+ 2CI \lp \tprod{\d u\otimes \d \omega}{\d u \otimes \d \omega}
				+\tprod{\d \omega\otimes \d u}{\d\omega\otimes \d u}
			 \rp\\
		&\qquad\qquad+ I^2 \d \omega_i \otimes \d\omega_j \otimes \d\omega^i \otimes \d\omega^j 
		\Big)\\
		&\quad +
		\frac{16}{u^4}
		I \Big( \d u \otimes \d u \otimes \tprod{\d \omega}{\d\omega}
			+ \d u \otimes \tprod{\d\omega}{\d\omega}\otimes \d u\\
			&\qquad\qquad+ \tprod{\d \omega \otimes \d u}{\d u \otimes \d \omega}
			+ \tprod{\d\omega}{\d\omega}\otimes \d u\otimes \d u
		\Big)\\
	&= \frac{16}{u^4}\bigg(
		4C \d u^{\otimes 4} + \frac{I^2}C \d \omega_i \otimes \d\omega_j \otimes \d\omega^i \otimes \d\omega^j\\
		&\quad+ I (\d u\otimes \d u)\dottimes \tprod{\d\omega}{\d\omega}
		+ I \big(\tprod{\d u\otimes \d \omega}{\d u \otimes \d \omega}
						+\tprod{\d \omega\otimes \d u}{\d\omega\otimes \d u} \big)
		\bigg)
\end{aligned}	
\]
We finally obtain:
	\begin{multline}
	\frac{u^4}{16} K = 10C \d u^{\otimes 4}
		- \frac12 \frac{I^2}C \d \omega_i \otimes \d\omega_j\otimes \d\omega^i\otimes \d\omega^j\\
		+ \frac12 I \Big(
			\d u \otimes \d u \otimes \tprod{\d \omega}{\d \omega}
			+ \tprod{\d \omega\otimes \d u}{\d u \otimes \d \omega}\\
			+ \d u\otimes \tprod{\d\omega}{\d\omega} \otimes \d u
			+ \tprod{\d\omega}{\d\omega}\otimes \d u \otimes \d u
		\Big)
	\end{multline}

\subsection{Generalisation to a general rigid body}

Although it will not be needed for the ideal gas, we identify here the Hessian geometry for a rotating rigid body without the assumption of spherical symmetry.

Without the symmetry assumption, however, the kinetic energy cannot be expressed as a function of the angular velocity which is independent of the configuration $\theta\in \SO_3$ of the body. It is nonetheless a function of the \emph{body-frame angular velocity} $\Omega:=\Ad_\theta^{-1} \omega \in \so_3$. The inertia tensor is then a general inner product in $\so_3$ and the kinetic energy takes the form
\[
	\Erot(\Omega) = \frac12 I(\Omega, \Omega)
\]

However, the momentum $\Omega$ is not conserved during a generic inertial trajectory. It can be understood as a momentum in the geometrical sense~\cite{SymGeoAnaMech}, but is not a physically conserved momentum. Therefore, the subsequent model is more of a geometrical construction than a description a physical system.

The analysis goes through similarly to the spherical symmetric case, with one important difference: the system is no longer invariant under the action of $\SO_3$ but it is instead invariant under the action of the isometry group $\SO(I)\subset \GL(\so_3)$.

The total energy is
\[
	E(T, \Omega) = \Uint + \Erot
\]
with
\begin{align}
	\Uint(T) &= CT\\
	\Erot(\Omega) &=\frac12 I(\Omega, \Omega)
\end{align}
The angular momentum is now
\[
	M = I(\Omega, \argdot) \in \so_3^*
\]

The same argument as in Section~\ref{secno:rigidbodyspherical} leads us to consider the inverse temperature $\beta = 1/T$ and the rotation vector $r = -\beta\Omega \in \so_3$ as dual coordinates to $E$ and $M$. The Hessian takes the form
\begin{equation}
	-\d \beta\otimes \d E  - \tprod{\d r}{\d M} 
\end{equation}
which can be computed likewise:
\[\begin{aligned}
	g := -\d \beta\otimes \d E  - \tprod{\d r}{\d M}
		&= \beta\lp
			\frac{C}\beta \frac{\d\beta}\beta \otimes\frac{\d\beta}\beta + I (\d \omega \otimes \d \omega) \rp
\end{aligned}\]
with $I(\d \omega\otimes \d \omega)$ denoting the contraction of the $\so_3\otimes \so_3$-valued bilinear form $\d \omega\otimes \d \omega$ with the product $I$.
Notice that the $\so_3$ component of the metric is the inertia tensor with a conformal factor of $\beta$. In particular the metric is indeed nondegenerate if and only if $I$ is nondegenerate. This is always the case when the body is non-planar, which we now assume.

The metric structure can manifestly be seen as the warped product of $(\so_3, I)$ by the temperature line $\Rpp$. The change of coordinates of Section~\ref{secno:rigidbodyspherical} has an obvious equivalent here. Writing $\sprod\argdot\argdot$ an $\ad$-invariant inner product on $\so_3$, we define $\sqrt{I}$ as the positive selfadjoint endomorphism of $\so_3$ such that
\[
	I(\Omega_1, \Omega_2) = \sprod{\sqrt I \Omega_1}{\sqrt I \Omega_2}
\]
and introduce the coordinates $\tilde\Omega = \sqrt I \Omega \in \so_3$. Using the coordinate $\tilde u= 2\sqrt{\Uint}$ the metric takes the standard form of the hyperbolic metric on the half space:
\begin{equation}
	\frac{4C}{\tilde u^2} \lp
			\d \tilde u\otimes \d \tilde u
			+ \tprod{\d \tilde\Omega}{\d \tilde \Omega} \rp	
\end{equation}
Since we assumed the inertia tensor to be non-degenerate, we can define a group morphism
\[
	G \in \SO(I) \mapsto \sqrt I G {\sqrt I}^{-1} \in \SO(\sprod\argdot\argdot)
\]
such that the map $\Omega\in \so_3 \mapsto \tilde \Omega \in \so_3$ is intertwining for the actions of $\SO(I)$ and $\SO(\so_3)$.

\paragraph{Poisson structure}

Without any symmetry assumption on $I$, the group actions of $\SO_3$ or $\SO(I)$ on $\so_3$ do not correspond to the coadjoint action of $\SO_3$ on $\so_3^*$ under the angular momentum map $M : \Rpp\times \so_3 \to \so_3^*$.

The group $\SO(I)$ acts by isometries of the metric $g$ but the Poisson structure (Section~\ref{secno:PoissonStructureGibbsSet}) defined by the coordinates $(E,M)$ does not make this action Hamiltonian. The symplectic leaves are the submanifolds of fixed $u$ (or $\beta$) and constant square norm $M^2$. We can reformulate this condition as follows, using on $\so_3^*$ the inverse invariant inner product to the one fixed on $\so_3$:
\[\begin{aligned}
	\sprod{M}{M}_{\so_3^*}
		&= \sprod{I(\Omega)}{I(\Omega)}_{\so_3^*}\\
		&= \sprod{\sqrt I^2(\Omega)}{\sqrt I^2(\Omega)}_{\so_3}\\
		&= \sprod{\sqrt I \tilde \Omega}{\sqrt I \tilde \Omega}_{\so_3}\\
		&= I(\tilde \Omega, \tilde \Omega)
\end{aligned}\]

Accordingly, the symplectic leaves, seen in coordinates $(u, \tilde\Omega)$, are the $2$-ellipsoids of constant $I$-norm and constant $u$.
\footnote{This may be reminiscent of kinetic energy but kinetic energy is actually $\frac12 I(\Omega, \Omega) = \frac12 \sprod{\tilde\Omega}{\tilde\Omega}$.}

According to Formula~\eqref{eqno:Poisson}, the Poisson bracket between two functions $f_1, f_2$ can be expressed using the tensorial bracket $[\argdot, \argdot]_{\setR\times\so_3}$ associated with the coordinates $(\beta, \Omega)$. Note that this bracket can also be expressed as a vector product associated to the scalar product $\sprod\argdot\argdot$ (which is a priori different from $g$).
\begin{equation}
	\{f_1, f_2\} := \sprod{M}{[\grad f_1, \grad f_2]_{\setR \times \so_3}} = I(\Omega, [\grad f_1, \grad f_2]_{\setR \times \so_3})
\end{equation}
Of course, since $\setR$ has the structure of an abelian Lie algebra, $\{f_1, f_2\}$ purely depends on the differentials of $f_1, f_2$ in the directions of constant $\beta$. 

\section{Ideal gas in rotation}\label{secno:PerfGas}

We now proceed to the system of an ideal gas in rotation.
It will be constructed as a statistical mechanical system according to Souriau's approach~\cite{SSDEng}. The corresponding mechanical system is that of a free point particle that is mathematically confined to a ball in Euclidean space. Since there is no interaction between the different point particles, the system will behave as an ideal gas.
\footnote{The appropriate modelling of an assembly of $N$ particle would use a distribution normalised to a total measure of $N$, with $N$ large in order for the statistical modelling to be relevant.
We will work with probability distributions, therefore the expected value of a quantity $X$ should be understood as \enquote{$X$ per particle}.}

Although the Riemannian geometry is more complex, we will show that in certain directions it can be asymptotically compared to the case of a rigid body.

\subsection{Mechanics of the massive free particle}

We first need to describe the Hamiltonian manifold corresponding to a massive free particle. Let $m\in \Rpp$ be a parameter that will correspond to the mass of the particle.
The manifold will be described as the tangent%
\footnote{It is more naturally described as a cotangent bundle.}
bundle 
$T \setR^3 \simeq \setR^3\times \setR^{3}$.
It has an action of the whole Galilean group~\cite{SSDEng} but we will only be interested in the action of the subgroup $G=\setR\times \SO_3$ with $\SO_3$ acting by rotations and $\setR$ being the $1$-parameter group describing inertial evolution (time translations).
We need to fix an origin in $\setR^3$ (it roughly corresponds to fixing the embedding of $\SO_3$ into the Galilean group).
Rotations $\mathcal R\in \SO_3$ then act simply as 
\[
	\mathcal R\cdot (q,v) = \lp \mathcal R \cdot q, \mathcal R \cdot v \rp
\]
and time translations act as
\[
	\exp(s\dd_{t})\cdot (q,v) = (q+sv, v)
\]

The symplectic structure is defined as
\[
	m \wprod{\d v}{\d q}
\]
with the contraction denoting the Euclidean scalar product (according to the notation described in Section~\ref{secno:notations}).

Let us write $(\tau, r)$ for a generic element of the Lie algebra $\Glie$.
The action of $G$ is Hamiltonian with the following momentum map:
\begin{equation}
	J(q,v) = m \lp \sprod{v}{\d r \cdot q} -\frac12 v^2 \d \tau \rp \in \Glie^*
\end{equation}
The covector $\d \tau$ is dual to a normalised generator $\dd_t$ of $\Lie(\setR)$ such that $\d \tau(\dd_t) = 1$. Physically, the component dual to $\d r$ is the angular momentum, and the opposite of the component dual to $\d \tau$ is the energy of the particle.

Since we will consider particles that are confined to a spherical subset of space, we will work with an open submanifold
$N = T B(0,R)$ where $B(0,R)\subset \setR^3$ is the open ball of radius $R>0$ centred at the origin.
Note that while this manifold is stable under the action of $\SO_3$ it is \textbf{not stable} under time evolution.
Indeed, inertial motion without further interaction is not compatible with the confinement to a bounded region of space.%
\footnote{This approach can be understood in the framework of scattering theory, as described in~\cite{SSDEng}.}
Nonetheless, this does not prevent the consideration of Gibbs distributions.
However, this model is not suitable for studying dynamics, and there is no guarantee the set of Gibbs distribution has a well defined group action of $\setR$.
Nevertheless, since $\setR$ is in the centre of $G$ or, more precisely, since its infinitesimal action on $N$ preserves $J$, it is natural to consider that $\setR$ acts trivially on the Gibbs set.

\subsection{Gibbs states of a spherically confined free particle}

Let us now identify the Gibbs states associated with $(N,J)$. Let $(\tau, r) \in \Glie$. We first consider the case $\tau = 0$: the integral which needs to converge is
\[
	\int_{(q,v)\in T B(0,R)}
		e^{m\sprod{v}{r\cdot q}} m^3|\d^3 v||\d^3 q| 
\]
The partial integral with respect to $q\in B(0,R)$ converges but the integral with respect to $v\in \setR^3$ never converges. Thus $(0, r)$ is never a generalised temperature.

We now assume that $\tau\neq 0$ and first compute

\[
	\sprod{(\tau, r)}{J(q,v)}
		= m \sprod{v}{r \cdot q} -\frac12 mv^2 \tau
		= -\tau m\lp \frac12 \lp v - \frac{r \cdot q}{\tau} \rp^2 - \frac12 \lp \frac{r \cdot q}{\tau} \rp^2 \rp	
\]
Defining $\beta=-\tau$ and $\omega = \frac{r}{\tau} = -\frac{r}{\beta}$ we can write
\[
	\sprod{(\tau, r)}{J(q,v)}
		= \beta m\lp \frac12 \lp v - \omega \cdot q \rp^2 - \frac12 \lp \omega \cdot q \rp^2 \rp	
\]
Note that $\omega$ is a quantity homogeneous to an angular velocity. Observe how the momentum looks formally like an energy which also depends on $\omega$. We look to compute the following integral:
\[
	Z(\beta, r) 
		:=
		\int_{(q,v)\in T B(0,R)}
		e^{-J(r, \tau)} 
		m^3|\d^3 q \wedge \d^3 v|
\]
Let us first factor it:
\[\begin{aligned}
	Z(\beta, r)
		&= 	\int_{(q,v)\in T B(0,R)}
			e^{- J(r, \tau)} 
			m^3|\d^3 q \wedge \d^3 v| \\
		&= \int_{(q,v)\in B(0,R)\times \setR^3}
			e^{-\beta m\lp \frac12 \lp v - \omega \cdot q \rp^2 - \frac12 \lp \omega \cdot q \rp^2 \rp}
			m^3 |\d^3 q|\otimes |\d^3 v| \\
		&= \int_{(q,v)\in B(0,R)\times \setR^3}
			e^{-\beta m\lp \frac12 v^2 - \frac12 \lp \omega \cdot q \rp^2 \rp}
			m^3 |\d^3 q|\otimes |\d^3 v| \\
		&= \underbrace{
					\int_{B(0,R)} e^{\beta m \frac12 \lp \omega \cdot q \rp^2}
					|\d^3 q|
					}_{\Zrot(\beta, r)}
		   \underbrace{
		   			\int_{\setR^3} e^{-\beta m\frac12 v^2}
					m^3 |\d^3 v|
					}_{\Zkin(\beta, r)}
\end{aligned}\]
We are looking for the set of $(\beta, r)$ such that both of these positive integrals converge.
The total integral can be understood as a \enquote{Gaussian of signature $(3,3)$} which is integrated over a cylindrical subset of $\setR^6$.
The second factor is easy to compute: it converges if and only if $\beta m>0$, namely $\beta >0$. This is a Gaussian integral which has a value independent of $\omega \cdot q$:
\begin{equation}\label{eqno:Zkin}
	\Zkin (\beta, r) = \lp \frac{2\pi m}{\beta} \rp^{\frac32}
\end{equation}

The integral defining $\Zrot$ has a bounded integrand and a domain of finite measure: it always converge.
Note however that the integrand is not integrable over the total Euclidean space, which shows that the whole Euclidean space admits no generalised temperature for the Galilean group.

We have therefore identified the Gibbs set, which we will describe as follows:
\[
	\Gibbs = \{(\beta, \omega)\}_{(\beta, \omega) \in \Rpp\times \so_3}
\]
Remember that the embedding into $\Glie$ is the following: $(\beta, \omega) \in \Gibbs \mapsto (\tau=-\beta, r=-\beta\omega)\in \Lie(\setR)\times \so_3$.
From now on, we stop using the coordinate $\tau$ (but we will occasionally use $r$).
It is important to note that $(\beta, \omega)$ does not form a set of \emph{flat coordinates} for the flat affine structure induced by the embedding into $\Glie$ (while $(\beta, r)$ does).

\subsection{Thermodynamic geometry of the Gibbs set}
\label{secno:GeoGibbsSet}

The Hessian metric we want to study is defined as the Hessian of $z=\ln(Z)$ on $\Gibbs$.
In order to compute this Hessian, we first need to compute the partial partition function $\Zrot(\beta,\omega)$.
It will be useful to consider the operator norm on $\so_3$, which is minus half the Killing form, so that if $\omega\in \so_3$ is nonzero and $q\in \setR^3$ belongs to the rotation plane then
\[
	(\omega\cdot q)^2 = \omega^2 q^2 \text{ with } \omega^2 := \lVert \omega \rVert^2
\]
Let us introduce cylindrical coordinates along the axis of rotation (the invariant line) of $\omega$ (any arbitrary axis if $\omega = 0$): they correspond to a map
\[
	(y, \theta, \rho) \in \setR \times \setS^1 \times \setR_+ \to \setR^3
\]
which maps the cylindrical measure $|\d y \wedge \d\rho \wedge \rho \d\theta|$ to the Euclidean measure $|\d^3 q|$. The following subset 
\[
	N = \{(y,\theta, \rho)\in [-R,R]\times \setS^1 \times [0,R] \suchthat y^2 + \rho^2 \leqslant R^2 \}\subset \setR \times \setS^1 \times \setR_+
\]
is mapped to $B(0,R)$.
The momentum takes the following form
\[
	\sprod{-\beta(1, \omega)}{J((y, \theta, \rho),v)}
		= \beta m\lp \frac12 \lp v - \omega \cdot q \rp^2 - \frac12 \omega^2 \rho^2  \rp	
\]
and we can compute the integral
\[\begin{aligned}
	\Zrot(\beta, \omega)
		&= \int_{N} e^{\beta m \frac12\omega^2 \rho^2}\rho \d y \d \rho \d\theta\\
		&= \int_{-R}^R \lp 
				2\pi \int_0^{\sqrt{R^2-y^2}} e^{\beta m \frac12\omega^2 \rho^2}\rho \d \rho
			\rp \d y\\
		&= \frac{4\pi}{\beta m \omega^2} \int_0^R 
				 e^{\beta m \frac12 \omega^2 (R^2-y^2)} - 1 \ \d y\\
		&= \frac{4\pi R}{\beta m \omega^2} \lp \int_0^1 
				 e^{\beta m \frac12 \omega^2 R^2(1-y^2)} \d y - 1 \rp		
\end{aligned}\]
Although it will not be of use to us, this \enquote{partial} Gaussian integral can be expressed using the lower incomplete Gamma function $\gamma: (x,y)\mapsto \int_0^y t^{x-1}e^{-t} \d t$ as follows:
\[
	\Zrot(\beta, \omega)
	= \frac{4\pi R}{\beta m \omega^2}
	\left[
		\frac{e^{\beta m\frac12\omega^2 R^2}}{\sqrt{2\beta m \omega^2 R^2}} \gamma\lp \frac12, \sqrt{\frac12\beta m \omega^2 R^2}\rp - 1 
	\right]
\]

Notice that $\Zrot$ depends purely on $\beta\omega^2 = \frac{r^2}{\beta}$.
In particular, along with Equation~\eqref{eqno:Zkin} this implies that $Z$ satisfies the following scaling property, for any $\eta\in \setR^*$:
\begin{equation}\label{eqno:Zetaomega}
	Z(\beta, \eta\omega) = |\eta^3| Z(\eta^2\beta, \omega)
\end{equation}

The Hessian metric of $\Gibbs$ is the Hessian of the potential $z = \ln(Z)$ in the flat coordinates $(\beta, r)$~\cite{HessianStructures}.
As mentioned in Section~\ref{secno:HessGeoLieThermo}, the differential of $z$ defines the expected energy $E$ and angular momentum $M$ as follows: 
\[
	\d z = -E \d \beta - \sprod{M}{\d r}
\]
The multiplicative decomposition of $Z$ gives an additive decomposition:
\[
	z = \zkin + \zrot
\]
with
\begin{align*}
	\zkin &= \ln(\Zkin)\\
	\zrot &= \ln(\Zrot)
\end{align*}

We can compute right away the Hessian of the first term:
\[
	\zint(\beta) = \frac32\ln \lp \frac{2\pi m}{\beta}  \rp
\]
that is readily differentiated:
\[
	\d \zint = -\Uint \d\beta = -\frac32\frac{\d\beta}{\beta}
\]
\[
	D^2 \zint = \frac32 \frac{\d\beta}{\beta^2}
\]
One can identify a heat capacity term with $C=\frac32$, corresponding to the tridimensional free motion of the gas particles.

The second term is more complicated. The potential term is
\[
	\zrot(\beta\omega^2) = \ln\lp \frac{4\pi R}{\beta m\omega^2} \rp +
		\ln \lp \int_0^1 
						 e^{\beta m \frac12 \omega^2 R^2(1-y^2)} - 1 \ \d y	
			\rp
\]
Let us define the function $I : \Rpp \to \setR$ such that
\begin{equation}\label{eqno:GasInertia}
	I(\beta\omega^2) = \frac{\dd \zrot(\beta\omega^2)}{\dd \frac12 \beta\omega^2}
\end{equation}
Since $\Zrot(\beta\omega^2) = \int_{B(0,R)} e^{\beta m \frac12 \lp \omega \cdot q \rp^2}
					|\d^3 q|$
is a strictly increasing function of $\beta\omega^2$, $I(\beta\omega^2) > 0$.
					
The differential of $\zrot$ can be expressed as 
\[
	\d \zrot
		= \frac12 I \d \lp \beta\omega^2 \rp
		= \frac12 I \d \lp \frac{r^2}\beta \rp
		= I \frac{2\beta \sprod{r}{\d r} - r^2 \d\beta}{2\beta^2}
		= -\frac12 I \omega^2 \d \beta  - I \sprod{\omega}{\d r}
\]
We can identify a contribution $\Erot(\beta, \omega) = \frac12 I(\beta\omega^2)\omega^2$ to the total energy
\[
	E = \Uint(\beta) + \Erot(\beta, \omega)
\]
and we obtain the expected angular momentum
\[
	M(\beta, \omega) = I \sprod{\omega}{\argdot} \in \so_3^*
\]
Notice that the function $I$ behaves formally as the inertia factor $I$ of Section~\ref{secno:rigidbodyspherical}, with its dependency on $\beta\omega^2$ reflecting the \enquote{internal structure} of the gas, as opposed to a rigid body.

We are interested in the Hessian of $z$ in the coordinates $(\beta, r)$.
It can be expressed as 
\[
	g = -\d E \otimes \d \beta - \tprod{\d M}{\d r}
\]
The metric can be expressed in several coordinates systems. Indeed, noticing that
\[
	\frac12 \d (\beta\omega^2)
		= \frac12 \d \lp \frac{r^2}{\beta} \rp 
		= \frac12 \lp -\frac{r^2}{\beta^2} \d \beta - \frac{2\sprod{r}{\d r}}{\beta} \rp
		= -\frac12\omega^2 \d \beta - \sprod{\omega}{\d r}
\]
we compute the Hessian
\[\begin{aligned}
	D^2 \lp \frac12 \beta\omega^2 \rp
		&= -\sprod{\omega}{\d \omega}\otimes \d\beta - \tprod{\d \omega}{\d r}\\
		&= - \tprod{\d \omega}{\lp \omega \d\beta + \d r \rp}\\
		&= \beta \tprod{\d \omega}{\d \omega}
\end{aligned}\]
which gives the following expression for the Hessian metric:
\[\begin{aligned}
	g 
		&= -\d \Uint\otimes \d \beta+ D\lp \frac12 I(\beta\omega^2)\d (\beta\omega^2) \rp\\
		&= -\d\Uint\otimes \d\beta + \beta I \tprod{\d\omega}{\d\omega} + \d I \otimes \d (\beta\omega^2)\\
		&= \frac32 \frac{\d\beta\otimes\d\beta}{\beta^2} + \beta I \tprod{\d\omega}{\d\omega} + I' \d (\beta\omega^2)\otimes \d (\beta\omega^2)
\end{aligned}\]
with $I' := \frac{\dd I}{\dd (\beta\omega^2)}$.
The expression of $g$ in the coordinates $(\beta, \omega)$ would be diagonal if it were not for the term $I'$.
This motivated studying the case of rigid bodies, which have a constant inertia tensor.
From a physical perspective, it is expected that at high velocity, the gas concentrates on the surface of the ball around a large circle that is orthogonal to the rotation axis. In particular, the coefficient $I$ is expected to be \enquote{asymptotically constant}, in a sense to be made precise. This matter will be the subject of Section~\ref{secno:highvelocity}.

Before studying the high velocity limit, let us describe the metric $g$ in the coordinates $(\beta, M)$.
This is the thermomechanical version of the coordinates introduced in~\cite{CurvatureMultiIdealGas}, and they give a \enquote{block-diagonalisation} of the metric. 
Indeed, the \emph{Maxwell relation}~\cite{BaezMaxwellRelations2}
\[
	\d^2 z = -\d E \wedge \d \beta - \wprod{\d M}{\d r} = 0
\]
implies that the terms of the metric in $\tprod{\d \beta}{\d M} + \tprod{\d M}{\d \beta}$ vanish. An alternative choice would be to use the coordinates $(E, r)$. Such coordinate sets are called \emph{mixed coordinates} and give a general way to locally split a Hessian manifold into two orthogonal and transversal foliations~\cite{AmaNagaInformationGeometry}.
\com{Effectuer tout de même le calcul afin de voir si la métrique ne prend pas une forme plus pratique.}

We need to express $\d r$ and $\d E$ in the coordinates $(\beta, M)$.
Let us first compute
\[\begin{aligned}
	\sprod{\argdot}{\d M}
		&= \d \lp -I \frac r \beta \rp \\
		&= - \frac r \beta I' \d (\beta \omega^2) + I \frac{r\d\beta - \beta \d r}{\beta^2}\\
		&= -\frac r\beta I' \frac{2\beta \sprod{r}{\d r} - r^2 \d \beta}{\beta^2} + I \frac{r\d\beta - \beta \d r}{\beta^2}\\
		&= \lp I' \frac{r^2 r}{\beta^3} + I \frac{r}{\beta^2} \rp \d \beta + \lp -\frac{2r^2}{\beta^2}I' - \frac1\beta I \rp \d r\\
	\sprod{\argdot}{-\beta \d M}
		&= \lp I' \beta \omega^2 + I\rp \omega \d \beta 
			+ \lp 2\beta\omega^2 I' + I \rp \d r 
\end{aligned}\]
We will explain in Section~\ref{secno:CurvAsympt} that $I'$ is always non-negative. Since $I>0$, this allows inverting the equation into
\begin{equation}
	\d r 
	= -\frac{1}{2\beta \omega^2 I' + I} \lp
			\beta \sprod{\argdot}{\d M} + \lp \beta\omega^2 I' + I \rp \omega \d \beta
		\rp
\end{equation}

We now express $\d E$ in terms of $\d \beta$ and $\d M$:
\[\begin{aligned}
	\d E 
		&= \d \Uint + \d \Erot\\
		&= \d \frac C \beta + \d \lp \frac12 \sprod{\omega}{M} \rp\\
		&= - \frac{C}{\beta^2}\d\beta
			+ \frac12\lp \sprod\omega {\d M }- \sprod M {\d \frac r \beta} \rp\\
		&= - \frac{C}{\beta^2}\d\beta
			+ \frac12\lp \sprod \omega {\d M} + \sprod M {\frac{r\d\beta - \beta\d r}{\beta^2}} \rp\\
		&= \lp \frac12 \sprod r M -C \rp \frac{\d\beta}{\beta^2}
			+ \frac12\sprod \omega{\d M}
			+ \frac1{2\beta} \frac{1}{2\beta \omega^2 I' + I}
				\lp \sprod M {\beta \d M}
					+ \lp \beta\omega^2 I' + I \rp \sprod{M}{\omega  \d \beta}
				\rp
		\\
		&= \lp -\frac12 \beta \sprod \omega M -C 
			+ \frac12 \beta \sprod M \omega
				\frac{\beta\omega^2 I' + I}{2\beta \omega^2 I' + I} \rp \frac{\d\beta}{\beta^2}
			+ \frac12 \lp \sprod{\omega}{\d M}
				+ \frac{\sprod{M}{\d M}}{2\beta \omega^2 I' + I} \rp
		\\
		&= - \lp C + \frac12 \beta \sprod M \omega \frac{\beta\omega^2I'}{2\beta\omega^2 I' + I} \rp 
			\frac{\d\beta}{\beta^2}
			+ \frac{\beta\omega^2 I' + I}{2\beta\omega^2I' + I} \sprod \omega {\d M}\\
		&= - \lp C + \frac{\beta\omega^2I'}{2\beta\omega^2 I' + I} \frac12 I \beta \omega^2\rp 
			\frac{\d\beta}{\beta^2}
			+ \frac{\beta\omega^2 I' + I}{2\beta\omega^2I' + I} \sprod \omega {\d M}\\
%
%
\end{aligned}\]
in which we momentarily used the dual inner product on $\so_3^*$.
We obtain the following expression 
\begin{equation}\label{eqno:gdbetadM}
	g = \lp C + \frac{\beta\omega^2I'}{2\beta\omega^2 I' + I} \frac12 I \beta \omega^2\rp 
				\frac{\d\beta\otimes\d\beta}{\beta^2}
			+ \frac1{2\beta\omega^2 I' + I}\beta \tprod{\d M}{\d M}
\end{equation}
which may prove be useful for a more detailed study of the properties of the Riemannian geometry. It can be compared to the metric of the rigid body as described in Equation~\eqref{eqno:RigidBodydM}, which has vanishing $I'$.
\com{Donner une expression pour la courbure scalaire ?}

\subsection{Inertia tensor}

We discuss here the inertia factor $I$ introduced in Equation \eqref{eqno:GasInertia}. By definition, the inertia tensor is an element of $I_{\beta, \omega}$ of $\Sym^2\lp \so_3^*\rp$ such that the angular momentum of the body is the element of $\so_3^*$ given by $I(\omega, \argdot)$. 

Given a generalised temperature $(\beta, \omega)\in \Gibbs$, there is a subgroup $\SO_2\subset \SO_3$ which stabilises $\omega$ and therefore stabilises the associated Gibbs distribution: Gibbs distributions have a cylindrical symmetry (spherical for $\omega = 0)$. But in general, the Gibbs distributions do not have spherical symmetry. Let us explain why the angular momentum has the same form as that of the spherically symmetric rigid body (Section~\ref{secno:rigidbodyspherical}).

The $\SO_3$-invariance of $\zrot$ implies that the angular momentum $M$ is as well, hence $I_{\beta, \omega}$ has to be invariant under rotations stabilising $\omega$, namely rotations around the same axis as $\omega$. This implies in particular that $M_{\beta, \omega} = I_{\beta, \omega}(\omega, \argdot)$ vanishes on the orthogonal of $\omega$ (for the standard invariant inner product) and, moreover, there is a factor $\alpha_{\beta, \omega}$ such that 
\[
	I_{\beta, \omega}(\omega, \argdot) = \alpha_{\beta, \omega}\sprod{\omega}{\argdot}
\]
In other words, under the usual isometry $\so_3 \simeq \setR^3$, the angular momentum is collinear to the rotation axis.
The invariance under the remainder of $\SO_3$ implies that $\alpha$ only depends on $\beta$ and $\omega^2$.

For this reason, even though the Gibbs states do not have spherical symmetry, the inertia tensor \enquote{follows} $\omega$ in such a way that the angular momentum and the kinetic energy behave as though the gas had spherical symmetry.

Since the partition function is constructed according to the formula
\[
	Z(\beta, r) = 
		\int_{(q,v)\in T B(0,R)}
		e^{\beta J(\omega, 1)} 
		m^3|\d^3 q \wedge \d^3 v|
\]
The angular momentum can be directly expressed, using ${}|_\beta$ to indicate partial derivative with fixed $\beta$, as follows:
\[\begin{aligned}
	\frac{1}{\beta} \frac{\dd z(\beta, r)}{\dd \omega}|_{\beta}
		&= \frac1\beta \frac{\dd \zrot(\beta\omega^2)}{\dd \omega}|_{\beta}\\
		&= \frac{\dd \zrot(\beta\omega^2)}{\dd \beta\omega^2}
		2\sprod{\omega}{\argdot}
\end{aligned}\]
This justifies defining $I$ as in Equation~\eqref{eqno:GasInertia}.

\subsection{Poisson structure}

The Poisson structure $\Lambda$ induced on $\Gibbs$ by the open embedding $\Gibbs \overset{Q = (E,M)}{\hookrightarrow} \setR\times \so_3^*$ can be expressed in a similar way to the case of the spherical rigid body (Section~\ref{secno:rigidbodyspherical}). Let us write $\Lambda_{\setR\times\so_3^*}$ for the canonical linear Poisson structure on $\setR\times\so_3$. We call $\Lambda_{\Rpp\times\so_3}$ the Poisson structure on $\Rpp\times\so_3$ which has for symplectic leaves the $2$-spheres with constant $\beta$, equipped with a $\SO_3$-invariant surface form of total surface $4\pi\rho$ if the sphere's $\so_3$-radius is $\rho$.

Since $\Lambda_{\setR\times\so_3^*}$ is $\SO_3$-invariant, so is $\Lambda$, so that it is proportional to $\Lambda_{\Rpp\times\so_3}$ on each leaf, with a factor depending on $\beta$ and $\omega^2$ (Proposition~\ref{propno:SympSph}). The sphere with Euclidean radius $\omega^2 = \rho^2$ has square radius $M^2 = I^2 \omega^2$ in $\setR\times\so_3^*$. As a consequence, its symplectic surface is $4\pi I \lVert \omega \rVert$ and the corresponding symplectic radius is $I \rho$. In conclusion,
\[
	\Lambda = \frac{1}{I(\beta\omega^2)} \Lambda_{\Rpp\times\so_3}
\]
Note that the factor $1/I(\beta\omega^2)$ is indeed constant on the $2$-spheres, since they have constant $\beta$ and constant $\omega^2$.

\subsection{Dependency in the radius}

Let us discuss the influence of the radius $R$ of the sphere on the Gibbs distributions. Only in this section do we explicitly write the functional dependency of the quantities on $R$.
Let $R_1, R_2$ be two positive real numbers with ratio $\eta := \frac{R_2}{R_1}$.
The dilation of factor $\eta$ maps diffeomorphically the ball $B(0,R_1)$ to $B(0,R_2)$; its differential defines a diffeomorphism $\psi : TB(0,R_1) \to TB(0,R_2)$ such that $\psi^*\wprod{\d v}{\d q} = \eta^2 \wprod{\d v}{\d q}$.
The Gibbs set's embedding in $\setR\times \so_3$, is independent of $R$.
Given $(\beta, \omega)\in \Gibbs$ we can express the partition function at radius $R_2$ as an integral on $TB(0,R_2)$:
\[\begin{aligned}
	Z(\beta, \omega, R_2)
		&= \lp \frac{2\pi m}\beta \rp^{\frac32} \int_{B(0,R_2)} e^{\frac12 \beta m (\omega \cdot q)^2} |\d^3 q|\\
		&= \lp \frac{2\pi m}\beta \rp^{\frac32} \int_{B(0,R_1)} e^{\frac12 \beta m \eta^2(\omega \cdot q)^2} \eta^3 |\d^3 q|\\
		&= \eta^{3} Z(\beta, \eta \omega, R_1)
\end{aligned}
\]
If we write $\mu_{\beta, \omega, R_i}$ the Gibbs densities, we have the equation
\[\begin{aligned}
	\psi^* \mu_{\beta, \omega, R_2}
	&= e^{-\frac12 \beta m\eta^2(v-\omega \cdot q)^2 + \frac12\beta m \eta^2(\omega\cdot q)^2} \frac{\eta^6|\d^3 q||\d^3 v|}{\eta^3 Z(\beta, \eta \omega, R_1) }\\
	&= e^{-\frac12 \beta m\eta^2(v-\omega \cdot q)^2 + \frac12\beta m \eta^2(\omega\cdot q)^2} \frac{|\d^3 q||\d^3 v|}{Z(\eta^2\beta, \omega, R_1) }
\end{aligned}\]
using the scaling Equation~\eqref{eqno:Zetaomega}.
We obtain the following relation:
\begin{equation}\label{eqno:muR1R2}
	\psi^* \mu_{\beta, \omega, R_2}
	= \mu_{\eta^2\beta, \omega, R_1}
\end{equation}
In conclusion, the Gibbs sets associated with balls of different radii can all be identified, with a change in the radius of the ball being equivalent to a change in the temperature.

\subsection{Decomposition of the Gibbs measures}\label{secno:GibbsDecomp}

Before discussing the high velocity limit of the Gibbs measures, let us describe the Gibbs set in more detail.
Given $(\beta, \omega) \in \Gamma$, the corresponding Gibbs measure $\mu_{\beta, \omega}$ takes the form
\begin{equation}\label{eqno:GibbsDecomp}
	\mu_{\beta, \omega}
		= \frac{e^{- J(r, \tau)} 
			m^3|\d^3 q \wedge \d^3 v|}{Z(\beta, \omega)} 
		= 
		\frac{e^{\beta m \frac12 (\omega\cdot q)^2 }
							|\d^3 q|}
		{\Zrot(\beta\omega^2)}		
		\otimes
		\lp \frac{\beta m}{2\pi} \rp^{3/2}
			{e^{-\beta m\frac12 \lp v - \omega \cdot q \rp^2}
				m^3 |\d^3 v|}		
\end{equation}
The Gibbs measure can be decomposed as a product of densities, between a density
\begin{equation}
	\nu_{\beta, \omega}:=\frac{e^{\beta m \frac12 (\omega\cdot q)^2 }
							|\d^3 q|}
		{\Zrot(\beta\omega^2)}
\end{equation}
which lives on $B(0,R)$ and depends purely on $\sqrt\beta \omega$, and a fibre density 
\[
	\lp \frac{\beta m}{2\pi} \rp^{3/2}
			{e^{-\beta m\frac12 \lp v - \omega \cdot q \rp^2}
				m^3 |\d^3 v|}
\]
which is a density on the manifold $T_q B(0,R)$, fibre of the tangent bundle $TB(0,R)\to B(0,R)$. In particular, the map
\[
	(\beta, \omega)\in \Gamma \mapsto \sqrt\beta \omega \in \so_3
\]
which defines a $\Rpp$-principal bundle fibration can be interpreted as a fibration of statistical models $\mu_{\beta, \omega}\mapsto \nu_{\beta, \omega}$.

Since the potential $z$ splits as $\zint(\beta) + \zrot(\beta\omega^2)$, one could hope for the fibration to be a Riemannian submersion, with the \enquote{horizontal metric} being $D^2 \zrot$. Indeed, the affine connection is invariant under the action of $\Rpp$, but it turns out that the vertical distribution is not invariant under the connection. This can be seen from the following expression, from Section~\ref{secno:GeoGibbsSet}:
\[
	D^2 \zrot
	= \beta I \tprod{\d\omega}{\d\omega}
	+ I' \d\lp \beta\omega^2 \rp\otimes \d \lp \beta\omega^2 \rp
\]
The term $\d(\beta\omega^2)\otimes \d(\beta\omega^2)$ can be pushed forward under $(\beta, \omega)\mapsto \sqrt\beta \omega$ but the term $\beta \tprod{\d\omega}{\d\omega}$ is not horizontal and therefore cannot be a pullback. The $\Rpp$-invariance of $D$ implies however that $D^2 \zrot$ is $\Rpp$-\emph{invariant}.

Furthermore, the measure $\nu_{\beta, \omega}$ is invariant when $\omega$ is replaced with $-\omega$. Since we will be interested in the behaviour of $\nu_{\beta, \omega}$ as $\beta\omega^2\to \infty$, we will have in mind the following alternative parametrisation:
\begin{equation*}
	\sqrt\beta\omega \in \so_3\smallsetminus\{0\}
		\mapsto
		(\beta\omega^2, [\omega])\in \Rpp\times \setP(\so_3)
\end{equation*}
This defines a twofold covering that exactly forgets the sign data from which $\nu_{\beta, \omega}$ is independent.

Although $\zrot$ is constructed as the logarithmic partition function of the exponential family of distributions $\nu_{\beta, \omega}$ which depend on the three dimensional parameter $\sqrt\beta\omega$, it turns out that it can be described as the logarithmic partition function of a \emph{different} family of distributions that effectively depends on a one-dimensional parameter. Let $(\beta, \omega)\in\Gibbs$ with adapted cylindrical coordinates $(y, \theta, \rho)$ on $B(0,R)$ (the axis can be arbitrary when $\omega=0$). We consider the application $\rho : q\in B(0,R)\mapsto \rho(q) \in [0,R)$. The image measure of $e^{\beta m \frac12(\omega\cdot q)^2}|\d^3 q|$ corresponds to the following density:
\[
	\rho_* \lp e^{\beta m \frac12(\omega\cdot q)^2}|\d^3 q| \rp
		= \lp \int_{\setS^1 \times (-\sqrt{R^2 -\rho^2}, \sqrt{R^2-\rho^2})} 1 \d y \d \theta \rp e^{\beta m \frac12\omega^2\rho^2} \d \rho
		= e^{\beta m \frac12\omega^2\rho^2} 4\pi \sqrt{R^2-\rho^2} \d \rho
\]
The total mass is exactly the partition function of $\nu_{\beta, \omega}$ thus the marginal distribution of the radial distance $\rho$ is given by
\[
	\rho_* \nu_{\beta,\omega}
		= \frac{e^{\beta m \frac12\omega^2\rho^2} \sqrt{R^2-\rho^2} \d \rho}
			{\int_{[0,R)} e^{\beta m \frac12\omega^2\rho^2} \sqrt{R^2-\rho^2} \d \rho}
\]
The image measures $\rho_* \nu_{\beta,\omega}$ therefore define an exponential family of probability distributions on $[0,R)$ that purely depends on $\beta\omega^2\in \Rpp$, namely it is independent of the direction of $\omega$. This family has a partition function that is identical to that of $\nu_{\beta, \omega}$ (up to an eventual constant factor). In other words, we have shown that $\zrot(\beta\omega^2)$ coincides with the logarithmic partition function of the family of probability distributions $\rho_*\nu_{\beta, \omega}$ which are effectively parametrised by the one-dimensional parameter $\beta\omega^2\in \Rpp$.

The relevance of the measure $\rho_* \nu_{\beta, \omega}$ will be clear in Section~\ref{secno:highvelocity} and the interpretation of $\zrot$ as a one-parameter logarithmic partition function will allow for simplifications in the calculations of its covariant derivatives (which will be the object of Section~\ref{secno:CurvedExp}).

\subsection{High angular velocity limit of the Gibbs distribution}\label{secno:highvelocity}

We want to investigate the behaviour of $g$ and its differentials as $\omega \to \infty$ with $\beta$ constant. The formulas (such as~\eqref{eqno:gdbetadM}) suggest investigating the behaviour of $I(\beta\omega^2)$ and its derivatives as $\beta\omega^2 \to \infty$. There is however a more geometrical approach: studying the asymptotic behaviour of the Gibbs measure as $\beta\omega^2\to \infty$. Indeed, the physically intuitive expectation is that it concentrates on a $1$-dimensional belt of radius $R$ in the plane of rotation. Since this belt is not a subset of the \emph{open ball} $B(0,R)$, such a result will require a compactification to the closed ball $\cBR$. 

Knowing the asymptotic behaviour of the Gibbs measure allows for estimating the asymptotic behaviour of its moments. Indeed, the partial integrals along the tangent spaces to $B(0,R)$ are standard Gaussian integrals, but we have no closed form for the integral on $B(0,R)$, essentially due to the integration domain. On the other hand, it can be proved that every iterated partial derivative of $g$ is a polynomial function of the moments $\EV[J^{\otimes k}]$ of the Gibbs measure. In particular, the Hessian curvature tensor of $g$ can be directly expressed as a rational function of $g$, $Dg$ and $D^2 g$ (as explained in Section~\ref{secno:HessGeo}) and can therefore be expressed as a rational function of the moments of order $2$ to $4$.
\com{Si on veut donner une expression exacte ici, il faut introduire plus tôt les familles exponentielles, et donner le lien entre leurs cumulants et leurs moments.}
In other words, the asymptotic behaviour of the Hessian curvature tensor can be directly related to the asymptotic behaviour of the Gibbs distribution, instead of working with the expression of the Hessian curvature for an arbitrary generalised temperature.

Rather than studying the limit of the Gibbs measure, we will adopt a mixed approach: we will use the exact contribution of $\zint$ to the Hessian geometry, but we will consider the limiting behaviour of $\nu_{\beta, \omega}$ and the asymptotic contributions of $\zrot$ to the Hessian geometry. Indeed, the expression $\zint = \lp \frac{2\pi m}{\beta} \rp^{\frac32}$ is very easy to handle, while the measure $e^{-\frac12 \beta m (v-\omega\cdot q)^2} \frac{|\d^3 v|}{\Zint}$ is centred at $\omega\cdot q$ which generically diverges as $\omega \to \infty$.

Given parameters $(\beta, \omega)\in \Gibbs$ with $\omega\neq 0$, and using the associated cylindrical coordinates for $B(0,R)$, the density $\nu_{\beta, \omega}$ takes the form
\begin{equation}
	\nu_{\beta, \omega}
	\frac{e^{\beta m \frac12 (\omega\cdot q)^2 }
							|\d^3 q|}
		{\Zrot(\beta\omega^2)}		
			= 
		\frac{e^{\beta m \frac12 \omega^2 \rho^2 }
							\rho\d\rho\d\theta\d y}
		{\Zrot(\beta\omega^2)}		
\end{equation}
Although the formula seemingly only depends on $\beta\omega^2$, one should keep in mind that the cylindrical coordinate system itself depends on $\omega$ and therefore the measure does depend on the direction of $\omega$ (a point of $\setP (\so_3)$).

Studying the limit of Gibbs measures is usually done with Laplace's method for approximating integrals. A situation similar to ours was considered in~\cite{LaplaceMethodProba} but they worked on the whole Euclidean space, while here an appropriate framework would be compact manifolds with boundary.
Our situation falls outside of the most standard case of application of Laplace's method because, as we will see, the amplitude vanishes at the maximum of the exponential term.

The proof will rely on the following asymptotic equivalent, which is the first order term given by Watson's Lemma~\cite{AsymptoticMethods, AppliedAsymptoticAnalysis}:
\begin{lemma}\label{lmno:LaplaceApprox}
	Let $A\in \Rpp$, $\alpha\in (-1, +\infty)$ and $F : [0,A]\to \setR$ a continuous function.
	Then 
	\[
		\lambda^{\alpha + 1} e^{-\lambda A} \int_0^A F(x)(A-x)^\alpha e^{\lambda x} \d x 
			\to_{\lambda \to +\infty}
		F(A) \Gamma(\alpha+1)
	\]
\end{lemma}
\begin{proof}
	It is a matter of integral manipulation:
	\[\begin{aligned}
		\int_0^A F(x)(A-x)^\alpha e^{\lambda x} \d x 
			&= 
				\int_0^A F(A-x) x^\alpha e^{\lambda (A-x)} \d x\\
			&=
				e^{\lambda A} \int_0^A F(A-x) x^\alpha e^{-\lambda x} \d x\\
			&=
				e^{\lambda A} \int_0^{\lambda A}
					F \lp A - \frac x \lambda \rp
					\frac{x^\alpha}{\lambda^\alpha} e^{-x} \frac{\d x}{\lambda}\\
			&=
				\frac{e^{\lambda A}}{\lambda^{\alpha+1}} \int_{\Rpp}
					\mathds{1}_{[0, \lambda A]}
					F \lp A - \frac x \lambda \rp
					{x^\alpha} e^{-x} {\d x}				
	\end{aligned}\]
	The integrand is uniformly dominated by $\lVert F \rVert_{\infty} x^{\alpha} e^{-x} \d x$
	and converges pointwise to $F(A)x^\alpha e^{-x} \d x$ as $\lambda \to +\infty$. Therefore
	\[
		{\lambda^{\alpha+1}} {e^{-\lambda A}} 
			\int_0^A F(x)(A-x)^\alpha e^{\lambda x} \d x 
		\to_{\lambda \to +\infty} 
		F(A) \Gamma(\alpha+1)
	\]	
\end{proof}

We can now give the high velocity limit of $\nu_{\beta, \omega}$ against functions which extend to $\cBR$:
\begin{theorem}[Weak limit of the Gibbs measures at high velocity]\label{thmno:GibbsLimit}
	Let $f : \cBR \to \setR$ be a continuous function and $\omega_0 \in \so_3\smallsetminus\{0\}$ used to define cylindrical coordinates. 
	Then
	\begin{equation}
		\int_{B(0,R)} f \nu_{\beta, \omega}
			\underset{
				\substack{
					\beta\omega^2\to \infty\\
					\omega \in \setR\omega_0}}
			{\longrightarrow}
			\int_{\{R\}\times \setS^1 \times \{0\}} f \frac{\d\theta}{2\pi}
	\end{equation}
	with $\{R\}\times \setS^1 \times \{0\}$ defining a cylindrical parametrisation of the large circle of radius $R$ in the rotation plane of $\omega_0$.
\end{theorem}

\begin{proof}
	Let us define the following average of $f$
	\[
		F : 
		\begin{cases}
			\rho\in [0,R[ &\mapsto \frac{1}{2\sqrt{R^2-\rho^2}}\int_{\setS^1\times [-\sqrt{R^2-\rho^2}, \sqrt{R^2-\rho^2}]} f(\rho, \theta, y) \d \theta\d y\\
			R&\mapsto \int_{\setS^1} f(R, \theta, 0) \d \theta 
		\end{cases}
	\]
	which is continuous and bounded.
	
	We want to identify the limit of the following integral:
	\[\begin{aligned}
		\int_{B(0,R)} f(\rho, \theta, y)  e^{\beta m \frac12 \omega^2 \rho^2} \rho\d \rho\d \theta\d y
			&= \int_0^R  \lp
				\int_{\setS^1\times [-\sqrt{R^2-\rho^2}, \sqrt{R^2-\rho^2}]} f(\rho, \theta, y)
					e^{\beta m \frac12 \omega^2 \rho^2} \d \theta\d y  \rp
				\rho \d \rho \\
			&= \int_0^R 2 \sqrt{R^2 - \rho^2} F(\rho) e^{\beta m \frac12 \omega^2 \rho^2} \rho \d \rho\\
			&\kern-0.5em\underset{\rho=x^2}{=} \int_0^{R^2} \sqrt{R^2 - x} F(\sqrt x) e^{\beta m \frac12 \omega^2 x} \d x\\
	\end{aligned}\]
	According to Lemma~\ref{lmno:LaplaceApprox} there is a convergence
\[\begin{tikzcd}
	\lp \frac{\beta m \omega^2} 2 \rp^{\frac32} e^{-\beta m \frac12\omega^2 R^2}
		\displaystyle\int_0^{R^2} \sqrt{R^2 - x} F(\sqrt x) e^{\beta m \frac12 \omega^2 x} \d x
	\ar[d, "\substack{\beta\omega^2 \to \infty\\ \omega\in \setR \omega_0}", shift right = 20]
	\\
		\Gamma\lp \frac32 \rp F(R)
		= \Gamma\lp \frac32 \rp \displaystyle\int_{\{R\}\times \setS^1 \times \{0\}} f {\d\theta}
\end{tikzcd}
\]
	The partition function $\Zrot$ is this very integral for $f=1$, and we can conclude
	\[
		\int_{B(0,R)} f \nu_{\beta, \omega}
		\xrightarrow[\substack{\beta\omega^2 \to \infty\\ \omega\in \setR \omega_0}]{}			
		\int_{\{R\}\times \setS^1 \times \{0\}} f \frac{\d\theta}{2\pi}
	\]	
\end{proof}
We see that the proof essentially proceeded by identifying the weak limit $\rho_* \nu_{\beta, \omega} \tobo \delta_R$ against continuous functions on $[0,R]$.


We now discuss the higher covariant derivatives of $z$ for a general family of probability distributions of the same form as $\rho_*\nu_{\beta, \omega}$.

\subsection{Cumulants of a curved exponential family of rank $1$}\label{secno:CurvedExp}\label{secno:CurvExp}
\com{Vérifier la terminologie !}

In this section, we derive general formulas for every higher differential of the logarithmic partition function $z$ for a class of probability distributions that have a similar form to $\rho_*\nu_{\beta, \omega}$ introduced in Section~\ref{secno:highvelocity}. This section works independently from the previous sections and is meant to state in generality the results we will need in Section~\ref{secno:CurvAsympt}.

Let $N$ be a topological space with a Borelian measure $\nu$.
Let $f : N \to \setR$ be a continuous \emph{bounded} function.
Let $\Gibbs$ be a smooth manifold with a smooth function $\theta : \Gibbs \to \setR$ which is everywhere regular%
\footnote{This hypothesis is not essential but allows for the convenient notation $\frac{\dd}{\dd\theta(\xi)}$.}
and such that for every $\xi \in \Gibbs$, the positive measure
\[
	e^{\theta(\xi) f} \nu 
\]
has finite total mass. Assume furthermore that $\Gibbs$ is equipped with a flat torsion-less affine connection $D$.

The logarithmic characteristic function is defined as follows:
\begin{equation}
	z(\xi) := \ln\lp \int_N 	e^{\theta(\xi) f} \nu \rp
\end{equation}
It is used to define the following family of probability measures
\[
	\nu_\xi := 	e^{\theta(\xi) f-z(\xi)} \nu 
\]
which will be called \enquote{curved exponential family of rank $1$}.
	
We are interested in the expression the higher covariant differentials of $z$.

\begin{lemma}\label{lmno:Cinf}
	The function $z$ is of class $\Cinf$ on $\Gibbs$.
\end{lemma}
\begin{proof}
	The function $(\xi, x)\in \Gibbs\times N \mapsto e^{\theta(\xi)f(x)}$ is differentiable at any order with respect to $\xi$.
	Any $n$-th derivative can be expressed as $P_\xi(f)e^{\theta(\xi)}$ with $P_\xi$ a polynomial function of degree $n$ with coefficients depending on $\xi$.
	As a consequence, it is uniformly integrable against $\nu$ in the neighbourhood of every $\xi\in \Gibbs$, which allows concluding that $e^z$, hence $z$, is of class $\Cinf$.
\end{proof}

\begin{lemma}
	For all integers $k$, the covariant tensor $D^k z \in \Gamma(T^* \Gibbs^{\otimes k})$ is totally symmetric.
\end{lemma}
\begin{proof}
	Since $D$ is curvature-free and torsion-free, for any tensor $T$ the following holds:
	\[\begin{aligned}
		D^2_{X,Y}T - D^2_{Y,X} T
			&= \lp D_XD_Y -D_{D_X Y} - D_Y D_X + D_{D_Y X} \rp T\\
			&= \lp [D_X, D_Y] -D_{[X,Y]} \rp T\\
			&= 0	
	\end{aligned}
	\]
	This allows a direct proof by recursion of the property: if $D^n z$ is totally symmetric then $D^{n+1} z$ is symmetric in the $2$ first entries, as well as in the $n$ last entries, which implies its total symmetry.
\end{proof}

\subsubsection{Faà di Bruno's formula}

Let us give a general formula for the iterated covariant derivative of a composite function $\Gibbs \to \setR \to \setR$.
There is a general formula known as \emph{Faà di Bruno}'s formula for the higher differentials of a composition of functions between vector spaces (under suitable topological hypotheses)~\cite{HighDerivativesComposite}. Since we are ultimately interested in the case when $(\Gibbs,D)$ is an affine manifold, we could justify using Fraenkel's formula as is. Let us however give an alternative recursive proof, that does not rely on the affine structure (using Taylor polynomials) but is formulated using the covariant differential.

Our proof of the covariant Faà di Bruno formula relies on the following version of the Leibniz formula:
\begin{lemma}\label{lmno:iXSymOp}
	Let $\Gibbs$ be a manifold, $\alpha\in \Omega^1(\Gibbs)$ and $D$ an affine connection on $\Gibbs$ such that for all integers $k$ the tensor $D^k \alpha$ is totally symmetric.
	
	
	For any two integers $n,j \geqslant1$ and any vector field $X$ the following formula holds:
	\begin{multline}\label{eqno:LeibnizSymOp}
		i_X\SymOp \lp
			\sum_{\substack{k : \llbracket 1, j +1\rrbracket\to\setN\\
					\sum k_i = n+1}}
				\frac{1}{\prod_i ((k_i+1) !)} 
				 D^{k} \alpha^{\otimes j+1}
				\rp\\
		= (j+1)\alpha(X) 
			\SymOp \lp
				\sum_{\substack{k : \llbracket 1, j \rrbracket\to\setN\\
						\sum k_i = n+1}}
					\frac{1}{\prod_i ((k_i+1) !)} 
			 D^{k} \alpha^{\otimes j}
					\rp
			+ D_X 
			\SymOp \lp
				\sum_{\substack{k : \llbracket 1, j +1\rrbracket\to\setN\\
						\sum k_i = n}}
					\frac{1}{\prod_i ((k_i+1) !)} 
			 D^{k} \alpha^{\otimes j+1}
						\rp
	\end{multline}
	with $D^{k}$ the differential operator $D^{k_1}\otimes D^{k_2}\dots \otimes D^{k_j}$.
\end{lemma}
\begin{remark}
	The combinatorial factor in the sum can be understood as follows.
	Given a partition $(k_i)_{1\leqslant i \leqslant j}$ of $n$ we define the \emph{increased} partial sums $\lp s_i = i + \sum_{l=0}^i k_l\rp_{0\leqslant i \leqslant j}$ so that the entries labelled by $\llbracket s_{i-1} + 1, s_i \rrbracket$ in $D^{k_1}\alpha \otimes \cdots \otimes D^{k_j} \alpha$ correspond to the factor $D^{k_i}\alpha$.
	Let us define the group $G_k\subset \mathcal S_{j+n}$ which is generated by
	$\Pi_{1\leqslant i \leqslant j} \mathcal S_{\llbracket s_{i-1} + 1, s_i \rrbracket}$
	and by the involutions exchanging $\llbracket s_{i-1} + 1, s_i\rrbracket$ with $\llbracket s_{j-1} + 1, s_{j} \rrbracket$ when $k_i=k_j$.
	The hypothesis that $D^k \alpha$ is totally symmetric for every $k$ implies that the tensor $D^k \alpha^{\otimes j}$ is invariant by permutation under $G_k$.
	
	The partitions which are a permutation of $(k_i)$ are exactly the $(k_{\sigma(i)})$ with $\sigma \in \mathcal S_j$. Their stabilisers in $\mathcal S_j$ have cardinal $\frac{|G_k|}{\prod_i ((k_i+1) !)}$ and these partitions form a set of cardinal $\frac{j! \prod_i ((k_i+1)!)}{|G_k|}$.
	The \enquote{minimal} symmetrising of $D^k \alpha^j$ (in a naive sense) would be
	\[\begin{aligned}
		\sum_{\sigma \in \mathcal S_{j+n}/G_k} 
			\sigma \cdot D^k \alpha^{\otimes j}
			&=
		\frac{1}{|G_k|}
		\SymOp D^k \alpha^{\otimes j}\\
			&= 
		\frac{1}{|G_k|}\SymOp
		\frac{|G_k|}{j!\prod_i ((k_i+1) !)}
			\sum_{u\in \mathcal S_j\cdot k}
			 D^u \alpha^{\otimes j}\\
			&= \frac{\SymOp}{j!} 
			\sum_{u\in \mathcal S_j\cdot k}
			\frac{1}{\prod_i ((u_i+1) !)}
			 D^u \alpha^{\otimes j}
	\end{aligned}\]
	In other words, up to a global factor of $j!$ the sums of Lemma~\ref{lmno:iXSymOp} are sums over the \emph{unsorted partitions} (partitions up to permutation) of naively symmetrised tensors.
\end{remark}
\begin{proof}
The proof will essentially rely on the reindexing of partitions. We will use the notation $i^{[m]}_X$ to denote the contraction with $X$ of the $m$-th entry of a covariant tensor, e.g. $i^{[1]}_X = i_X$. We start with the following commutation formula which holds on covariant $(n+j+2)$-tensors:
\[
	i_X \circ \SymOp = \SymOp \circ 
					{\sum_{m=1}^{n+j+2} i_X^{[m]}}
\]
Let $(k_i)$ be any $(j+1)$-partition of $n+1$. Writing $s_i$ for the associated partial sums $i + \sum_{l=0}^i k_l$, we define $\lambda_i = s_{i-1}+1$. Note the equivalence between $\lambda_i = \lambda_{i+1}$ and $k_i = 0$.
We can decompose the sum
\[\begin{aligned}
	\sum_{m=1}^{n+j+2} i_X^{[m]}
			 D^{k} \alpha^{\otimes j+1}
	&= 
	\sum_{m\suchthat\exists l, m=\lambda_l=\lambda_{l+1}} i_X^{[m]}
				 D^{k} \alpha^{\otimes j+1}
	+
	\sum_{m\suchthat\not\exists l, m=\lambda_l=\lambda_{l+1}}  i_X^{[m]}
				 D^{k} \alpha^{\otimes j+1}\\
	&= 
	\sum_{l\suchthat k_l=0} i_X^{[\lambda_l]}
				 D^{k} \alpha^{\otimes j+1}
	+
	\sum_{l\suchthat k_l\neq 0} (k_l+1) i_X^{[\lambda_l]}
				 D^{k} \alpha^{\otimes j+1}\\
	&= 
	\sum_{l\suchthat k_l=0}
			\alpha(X) 
				 D^{\check k} \alpha^{\otimes j}
	+
	\sum_{l\suchthat k_l\neq 0} (k_l+1) i_X^{[\lambda_l]}
				 D^{k} \alpha^{\otimes j+1}\\
\end{aligned}\]
writing $\check{k}$ for the $j$-partition of $n$ which is $k$ without the vanishing rank $l$ term (it is dependent on $l$).
 
We identify the sum over the partitions of the first term:
\[
\begin{aligned}
	\sum_{\substack{k : \llbracket 1, j +1\rrbracket\to\setN\\
				\sum k_i = n+1\\
				l\suchthat k_j=0}}
		\frac{\alpha(X)}{\prod_i ((k_i+1) !)} 
		 D^{\check k} \alpha^{\otimes j+1}
	&= 
	\sum_{\substack{\check{k} : \llbracket 1, j\rrbracket\to\setN\\
				\sum {\check k}_i = n+1\\
				0 \leqslant l\leqslant j}}
		\frac{\alpha(X)}{\prod_i (\check{k}_i+1) !} 
	 D^{\check k} \alpha^{\otimes j+1}\\
	&= (j+1)	\alpha(X) 
		\sum_{\substack{\check{k} : \llbracket 1, j\rrbracket\to\setN\\
				\sum \check{k_i} = n+1}}
			\frac{1}{\prod_i (\check{k}_i+1) !} 
		 D^{\check k} \alpha^{\otimes j+1}
\end{aligned}
\]
For the second term, we write $\tilde k_i$ for the partition with the $l$-th entry reduced by one and we have the following substitution
\[
\begin{aligned}
	\sum_{\substack{k : \llbracket 1, j +1\rrbracket\to\setN\\
				\sum k_i = n+1\\
				j\suchthat k_j\neq 0}}
			\frac{(k_l+1) i_X^{[\lambda_l]}}
			{\prod_i ((k_i+1) !)}
			 D^{k} \alpha^{\otimes j+1}	
	&= 
	\sum_{\substack{\tilde k : \llbracket 1, j +1\rrbracket\to\setN\\
				\sum \tilde k_i = n\\
				1 \leqslant l \leqslant j +1}}
			\frac{i_X^{[\lambda_l]}}
			{\prod_i (\tilde k_i+1) !}
			 D^{(\tilde k_i + \delta_{il})} \alpha^{\otimes j+1}	
	\\
	&= 
	\sum_{\substack{\tilde k : \llbracket 1, j +1\rrbracket\to\setN\\
				\sum \tilde k_i = n}}
			\frac{D_X}
			{\prod_i (\tilde k_i+1) !}
			 D^{\tilde k} \alpha^{\otimes j+1}		
\end{aligned}
\]
using the standard Leibniz rule for the covariant derivative of a tensor product.
In order to obtain the second term of Equation~\eqref{eqno:LeibnizSymOp}, it is enough to realise that $D_X$ commutes with $\SymOp$, which allows to conclude.
\end{proof}

The expression for $D^n z$ will be made more concise by using the multinomial coefficients, which we now introduce:
\begin{definition}[Multinomial coefficients]
	Given an ordered partition $(k_i)_{1\leqslant i \leqslant j}$ of an integer $n$, the number of $\llbracket 1, j \rrbracket$-indexed partitions of $\llbracket 1, n \rrbracket$ with the $i$-th part of size $k_i$ is denoted by the following \emph{multinomial coefficient}:
	\[
		\binom{n}{k_1, \ \dots \ , k_j}
		=
		\frac{n!}{\prod\limits_{1\leqslant i \leqslant j} (k_i !)}
	\]
\end{definition}

\begin{proposition}\label{propno:Dnz}
	Let $\Gibbs$ be a smooth manifold equipped with a connection $D$, $\U$ be an open subset of $\setR$ and $\theta : \Gibbs\to \U$ and $\zeta : \U \to \setR$ smooth functions.
	Define $z = \zeta\circ \theta$ and assume that for every integer $n$ the tensor $D^n z$ is totally symmetric.
	Let us adopt the notation $\frac{\dd^k z}{\dd \theta(\xi)^k}$ for the function $\zeta^{(k)}\circ \theta$ where $\zeta^{(k)}$ is the standard iterated derivative of $\zeta$.
	
	Then the iterated covariant differential of $z$ of any order $n\geqslant 1$ is given by the following formula:
	\begin{equation}
		D^n z
		= \sum_{j=1}^{n}
			\frac{1}{j!}\frac{\dd^j z}{\dd \theta(\xi)^j}
			\frac{\SymOp}{n!} \lp
				\sum_{\substack{k : \llbracket 1, j\rrbracket \to\setN\\
						\sum k_i = n-j}}
					\binom{n}{k_1+1, \ \dots\ , k_j+1} 
			 D^{k} (\d\theta^{\otimes j})
					\rp
	\end{equation}
\end{proposition}
\begin{proof}
	The proof proceeds by recursion. The case $n=1$ is straightforward. 
	
	Let $X$ be a vector field on $\Gibbs$ and assume that the formula holds for a given $n$.
		
	We compute $i_X D^{n+1} z = D_X D^n z$, using the expression of the multinomial coefficient in terms of factorials:
	\[
	\begin{aligned}
		&D_X D^n z 
			= 	D_X \sum_{j=1}^{n}
				\frac{\dd^j z}{\dd \theta(\xi)^j}
				\frac{\SymOp}{j!} \lp
						\sum_{\substack{k : \llbracket 1, j\rrbracket \to \setN\\
							\sum k_i = n-j}}
							\frac{1}{\prod_i ((k_i+1) !)} 
						 	D^{k} (\d\theta^{\otimes j})
							\rp\\
			&= 	\sum_{j=1}^{n}
				D_X	\lp \frac{\dd^j z}{\dd \theta(\xi)^j} \rp
				\frac{\SymOp}{j!} \lp
						\sum_{\substack{k : \llbracket 1, j\rrbracket\to \setN\\
							\sum k_i = n-j}}
							\frac{1}{\prod_i ((k_i+1) !)} 
						 	D^{k} (\d\theta^{\otimes j})
							\rp
				+
				\frac{\dd^j z}{\dd \theta(\xi)^j}
				\frac{D_X\SymOp}{j!} \lp
						\sum_{\substack{k : \llbracket 1, j\rrbracket\to \setN\\
							\sum k_i = n-j}}
							\frac{1}{\prod_i ((k_i+1) !)} 
						 	D^{k} (\d\theta^{\otimes j})
							\rp\\
			&= 	\sum_{j=1}^{n}
				\frac{\dd^{j+1} z}{\dd \theta(\xi)^{j+1}} X(\theta)
				\frac{\SymOp}{j!} \lp
						\sum_{\substack{k : \llbracket 1, j\rrbracket\to \setN\\
							\sum k_i = n-j}}
							\frac{1}{\prod_i ((k_i+1) !)} 
						 	D^{k} (\d\theta^{\otimes j})
							\rp
				+
				\frac{\dd^j z}{\dd \theta(\xi)^j}
				\frac{D_X\SymOp}{j!} \lp
						\sum_{\substack{k : \llbracket 1, j\rrbracket\to \setN\\
							\sum k_i = n-j}}
							\frac{1}{\prod_i ((k_i+1) !)} 
						 	D^{k} (\d\theta^{\otimes j})
							\rp\\
			&= 	\frac{\dd^{n+1} z}{\dd \theta(\xi)^{n+1}}
					X(\theta)(\d\theta^{n})
				+ \frac{\dd z}{\dd \theta(\xi)}
				D_X D^{n}\theta
				\\ & \qquad
	\begin{multlined}
				+\sum_{j=2}^{n}
				\frac{\dd^j z}{\dd \theta(\xi)^j}
				\left[
				X(\theta) 
				\frac{j\SymOp}{j!} \lp
						\sum_{\substack{k : \llbracket 1, j-1\rrbracket\to \setN\\
							\sum k_i = n-j+1}}
							\frac{1}{\prod_i ((k_i+1) !)} 
						 	D^{k} (\d\theta^{\otimes (j-1)})
							\rp
				+ \right.\\ \left.
			\frac{D_X\SymOp}{j!} \lp
						\sum_{\substack{k : \llbracket 1, j\rrbracket\to \setN\\
							\sum k_i = n-j}}
							\frac{1}{\prod_i ((k_i+1) !)} 
						 	D^{k} (\d\theta^{\otimes j})
							\rp
				\right]	
				\end{multlined}
				\\
			&= \sum_{j=1}^{n+1}
				\frac{\dd^j z}{\dd \theta(\xi)^j}
					i_X 
					\frac{\SymOp}{j!} \lp
						\sum_{\substack{k : \llbracket 1, j\rrbracket \to\setN\\
							\sum k_i = n+1-j}}
							\frac{1}{\prod_i ((k_i+1) !)} 
						 	D^{k} (\d\theta^{\otimes j})
							\rp
	\end{aligned}
	\]
	The last identity used Lemma~\ref{lmno:iXSymOp}.
	This proves that Equation~\eqref{eqno:Dnz} holds at order $n+1$ and concludes the proof.
\end{proof}

Proposition~\ref{propno:Dnz} has a higher rank version: one simple way to obtain it is to replace $\theta$ by a map to a vector space $\setR^m$ and $f$ by a map to the dual vector space ${\setR^m}^*$.

\subsubsection{Cumulants}

In the case at hand, $z$ is the logarithm of a partition function and its iterated derivatives are known as the \emph{cumulants}:
\begin{definition}[Cumulants]
Let $n$ be a positive integer. The \emph{$n$-th cumulant} of $f$ is defined as follows:
\begin{equation}
	c_n : \xi\in \Gibbs \mapsto \frac{\dd^n z(\xi)}{\dd \theta (\xi)^n}
\end{equation}
\end{definition}
\begin{example}
	The first cumulant is the expected value of $f$, according to the following calculation:
		\begin{equation}\label{eqno:dz}
			\d \ln \lp \int_N e^{\theta f} \nu \rp
			= \frac{\int_N f e^{\theta f} \nu}{\int_N e^{\theta f} \nu} \d \theta
			= \EV_\xi[f] \d\theta
		\end{equation}
	The second cumulant is the variance of $f$.
\end{example}

We now want to express the cumulants as expected values.
It will be convenient to use the following notation for the expected value:
\begin{equation}
	\EV_\xi : g \mapsto \EV_\xi[g] = \int_N g \nu_\xi
\end{equation}
The measure $\nu_\xi$ is differentiable with respect to $\xi$ in the following sense:
\begin{lemma}\label{lmno:dExi}
	Define on $N \times \Gibbs$ the following function: 
	\begin{equation}
		\hf(x, \xi) := f(x) - \EV_\xi[f]
	\end{equation}
	The function $\hf$ is continuous on $N \times \Gibbs$ and for every $x\in N$ the partial function
	$\xi\in \Gibbs \mapsto f(x) - \EV_\xi[f]$ is of class $\Cinf$ on $\Gibbs$.
	
	Furthermore for any bounded continuous function $g$ on $N$, the function
	$\xi\in \Gibbs \mapsto \EV_\xi[g]$ is of class $\Cinf$ and its differential takes the form
	\begin{equation}
		\d \EV_\xi[g] = \EV[\hf g]\d \theta(\xi)
	\end{equation}
\end{lemma}
\begin{proof}
	Let $g$ be any bounded continuous function on $N$.
	Since the iterated partial derivatives of $g(x)e^{\theta(\xi)f}$ with respect to $\xi$ are all bounded on $N$, with a bound locally uniform in $\xi$ (see the proof of Proposition~\ref{lmno:Cinf}), the parametrised integral
	\[
		\xi\in \Gibbs \mapsto \int_N g e^{\theta(\xi)f} \nu
	\]
	defines a $\Cinf$ function. In particular
	\[
		\EV_\xi[g] = e^{-z(\xi)} \int_N g e^{\theta(\xi)f} \nu
	\]
	is of class $\Cinf$. This holds in particular for $g=f$, which implies that the partial functions $\hf(\xi, \argdot)$ are of class $\Cinf$ for any $x\in N$.
	The differential of the parametrised integral is readily calculated:
	\[
		\d \lp \int_N g e^{\theta(\xi)f} \nu \rp =
		\lp \int_N g f e^{\theta(\xi)f} \nu \rp \d \theta
	\]
	and using Formula~\eqref{eqno:dz} one obtains
	\[
		\d e^{-z(\xi)} = \EV[f]e^{-z(\xi)} \d \theta
	\]
	which allows concluding the proof by applying the Leibniz formula to differentiate $\EV_\xi[g]$.
\end{proof}

Since we will have to consider expected values of functions parametrised by $\xi$, Lemma~\ref{lmno:dExi} needs to be generalised as follows:
\begin{lemma}\label{lmno:dExigen}
	Let $A$ denote the subset of functions $g: N\times \Gibbs \to \setR$ which satisfy the following set of properties:
	\begin{enumerate}
		\item $g$ is continuous on $N\times \Gibbs$,
		\item $g$ depends only on $x$ and $\theta(\xi)$,
		\item for every $x\in N$, the function $g(x, \argdot)$ is $\Cinf$,
		\item the partial derivatives of $g$ at every order are continuous on $N\times \Gibbs$,
		\item for every $\xi \in \Gibbs$ and $m\in \setN$ there exists a neighbourhood $\U$ of $\xi$ such that $\frac{\dd^m g}{\dd \theta(\xi)^m}$ is bounded on $N \times \U$.
	\end{enumerate}
	
	Then $A$ is an algebra stable under the operation $\frac{\dd}{\dd \theta(x)}$ and for all $g\in A$ the following property holds:
	\[
		\d \EV_\xi[g] = \EV_\xi \left[\hf g + \frac{\dd g}{\dd \theta(\xi)}\right]\d \theta(\xi)
	\]
\end{lemma}
The proof is straightforward.
We will use the following recursion property in order to express the cumulants:
\begin{proposition}\label{propno:fn}
	We define by recursion the following sequence of functions on $\Gibbs\times N$: 
	\[
	\begin{cases}
		f_1 := f\\
		f_2 := \hf^2\\
		f_{n+1} := \hf f_n + \frac{\dd f_n}{\dd \theta(\xi)} \qquad \text{for } n\geqslant 2
	\end{cases}
	\]
	For every $n$, $f_n$ belongs to the algebra $A$ defined in Lemma~\ref{lmno:dExigen}.
	Moreover, they satisfy the following relations:
	\[
		\forall n \in \setN, n\geqslant 1, \,
		\forall\xi \in \Gibbs, \,
			\EV_\xi[f_n] = c_n(\xi)
	\]
\end{proposition}

\begin{proof}
	It is clear that $f$ and $\hf$ belong to $A$ therefore a straightforward recurrence implies that $f_n\in A$ for every $n$.
	
	We have already stated that $c_1(\xi) = \EV_\xi[f]$.
	Since by definition $c_{n+1} = \frac{\dd c_n}{\dd \theta(\xi)}$, the relation will be proved by differentiating the parametric integral $\EV_\xi[f_n]$ according to Lemma~\ref{lmno:dExigen}:
	\[
		\d \EV_\xi[f_n]
			= \EV_\xi\left[ \hf f_n + \frac{\dd f_n}{\d\theta(\xi)} \right] \d\theta(\xi)
	\]	
	For $n\geqslant 2$ this gives 
	$\frac{\dd \EV_\xi[f_n]}{\dd\theta(\xi)} = \EV_\xi\left[ f_{n+1} \right]$. 
	For $n=1$, since $\EV_\xi[\hf]=0$ we can replace
	\[
		\EV_\xi\left[ \hf f\right] \d\theta(\xi)
			= \EV_\xi\left[ \hf^2\right] \d\theta(\xi)
			= c_2\d\theta(\xi)
	\]
	This allows concluding by recursion.	
\end{proof}
\begin{remark}
	It may be more natural to define the sequence which has the cumulants as expected values using the recursion relation $f_{n+1}:= \hf f_n + \frac{\dd f_n}{\dd \theta(\xi)}$ from $n=1$ but it will be more convenient to have expressions only involving $\hf$.
\end{remark}
We now obtain a version of the Faà di Bruno formula with expected values as coefficients:
\begin{theorem}\label{thmno:DnzEVfj}
	The iterative covariant derivatives of $z$ can be expressed as follows:
	\begin{equation}\label{eqno:Dnz}
		D^n z
		= \sum_{j=1}^{n}
			\frac{1}{j!}\EV_\xi[f_j]
			\frac{\SymOp}{n!} \lp
				\sum_{\substack{k : \llbracket 1, j\rrbracket \to\setN\\
						\sum k_i = n-j}}
					\binom{n}{k_1+1, \ \dots\ , k_j+1} 
			 D^{k} (\d\theta^{\otimes j})
					\rp
	\end{equation}
	with $f_j$ as defined in Proposition~\ref{propno:fn}.
\end{theorem}
Theorem~\ref{thmno:DnzEVfj} can be adapted for curved exponential families of higher rank. The expectation values of the $f_j$'s then take value in tensor products of ${\setR^m}^*$ while the tensors $D^k \d \theta$ take value in $\setR^m$ which allows for natural contractions.

\begin{example}[Differentials of $z$ of order up to $4$]\label{exno:D1234z}
	At low order, we have the expressions
	\begin{subequations}\label{eqnos:D1234z}
	\begin{align}
		D 	z	&= \EV_\xi[f] D \theta\\
		D^2 z	&= \EV_\xi[\hf^2] D\theta\otimes D\theta + \EV_\xi[f] D^2\theta\\
		D^3 z 	&= \EV_\xi[\hf^3]D \theta^{\otimes 3}
				  +	\EV_\xi[\hf^2] \frac{\SymOp}2 \lp D^2 \theta \otimes D \theta \rp
				  + \EV_\xi[f] D^3 \theta
		\\
		\begin{split}
		D^4	z	&= \lp \EV_\xi[\hf^4] - 3\EV_\xi[\hf^2]^2 \rp D\theta^{\otimes 4}
			  +	\EV_\xi[\hf^3] \frac{\SymOp}{8}
			  	\lp D^2 \theta \otimes D \theta \otimes D\theta \rp
			  \\ &\qquad\qquad
			  +	\EV_\xi[\hf^2] \SymOp \lp
			  		\frac{D^3\theta\otimes \d\theta}{6}
			  		+ \frac{D^2\theta\otimes D^2 \theta}{8} \rp
			  + \EV_\xi[f] D^4 \theta
  		\end{split}
	\end{align}
	\end{subequations}
\end{example}

These formulas are meant to be used to determine the high-velocity asymptotics of the Hessian curvature of the Gibbs set of the spherically confined ideal gas.
However, this is not a rank $1$ curved exponential family as discussed in the present section, but we explain why we can work as if it were.

\subsection{Curvature asymptotics at high velocity}\label{secno:CurvAsympt}

We propose to identify the asymptotic behaviour of the curvature tensor of $\Gibbs$ as $\beta\omega^2\to \infty$. More precisely, we will study the asymptotic behaviour of the \emph{Hessian curvature tensor} $K$ which was introduced in Section~\ref{secno:HessGeo}.

%
%
%
%

In order to obtain the asymptotic for $K$ we need to have the asymptotics for $g$, $D g$ and $D^2 g$, namely $D^2 z$, $D^3 z$ and $D^4 z$.
Let us start with the contribution from $\zint$: this is a direct calculation.
Recall the function $u = 2\beta^{-\frac12}$.
\begin{lemma}\label{lmno:D234zint}
	For every integer $n$,
	\begin{equation}\label{eqno:Dnzint}
		D^n \zint = (n-1)! \frac32 \beta^\frac{n/2} \d u^{\otimes n}
	\end{equation}
	In particular,
	\begin{subequations}\label{eqsno:D234zint}
	\begin{align}
		D^2 \zint &=  \frac32 \beta \d u \otimes \d u\\
		D^3 \zint &=  3 \beta^{3/2} \d u\otimes \d u\otimes \d u\\
		D^4 \zint &=  9 \frac32 \beta^2 \d u\otimes \d u\otimes \d u\otimes \d u
	\end{align}
	\end{subequations}
\end{lemma}
%

We want to use the results from Section~\ref{secno:CurvedExp} in order to obtain the iterated covariant differentials of $\zrot$. According to Section~\ref{secno:GibbsDecomp}, $\zrot$ coincides (up to an eventual constant) to the partition function of the rank $1$-curved exponential family which we wrote $\rho_*\nu_{\beta, \omega}$. 
Let us adopt the notation
\[
	\Ebo [f] := \int_{B(0,R)} f \nu_{\beta, \omega}
\]
For $\omega\neq 0$, we define the following function, which hence has an implicit dependency on $\omega$:
\[
	\iota : q \in B(0,R) \mapsto \frac12 m \rho^2 = \frac12 m \frac{(\omega\cdot q)^2}{\omega^2}
\]
We also introduce $\hi = \iota - \Ebo[\iota]$. These two functions' dependency on $q$ factors through $\rho : B(0,R)\to [0,R)$.
We deduce immediately from Section~\ref{secno:CurvedExp} that $\frac{\dd^2 \zrot}{\dd (\beta\omega^2)^2} = \Ebo[\hi^2] \geqslant 0$ which proves the claim in Section~\ref{secno:GeoGibbsSet} that $I' = 2 \frac{\dd^2 \zrot}{\dd (\beta\omega^2)^2}$ is non-negative.
We will use the formulas from Example~\ref{exno:D1234z} to obtain expression for the iterative covariant differentials of $\zrot$ up to order $4$.
Let us first identify the iterated covariant differentials of $\beta\omega^2$:
\begin{lemma}\label{lmno:Dnbo}
	Let $D$ denote the flat derivative in the coordinates $(\beta, r)$. Then
\begin{equation}
		D^2 \omega = -\frac1\beta \d\beta \dottimes\d\omega
		\label{eqno:Ddom}
\end{equation}
	and for every integer $n\geqslant 0$, 
	\begin{equation}\label{eqno:Dnbo}
		D^{n+2} (\beta\omega^2)
			= \frac{(-1)^n}{\beta^{n-1}} \d \beta^{\dottimes n}\dottimes \dprod{\d\omega}{\d\omega}
	\end{equation}
	with $\dottimes$ standing for the symmetrised product of covectors (see Section~\ref{secno:notations}).
\end{lemma}
%
\begin{proof}
	Identity~\eqref{eqno:Ddom} is obtained by the following calculations:
	\[
		D \omega = -\d\lp \frac r \beta \rp = \frac{r \d \beta - \beta \d r}{\beta^2} = -\frac1\beta \d r - \frac{\omega}{\beta} \d \beta
	\]
	and
	\[\begin{aligned}
		D^2 \omega
			&= -\d \lp \frac 1 \beta \rp \otimes \d r - \d \lp \frac \omega\beta \rp \otimes\d \beta\\
			&= -\frac{\d\beta \otimes (\omega \d\beta + \beta \d \omega)}{\beta^2}
			+ \frac{\omega\d\beta - \beta\d\omega}{\beta^2}\otimes \d\beta\\
			&= -\frac1\beta \lp \d \beta \otimes \d\omega + \d \omega \otimes \d \beta \rp
	\end{aligned}
	\]
	
	We have proved Equation~\eqref{eqno:Dnbo} for $n=0$ in Section~\ref{secno:GeoGibbsSet}. We complete the proof by recursion, assuming it holds for some integer $n$.	
	Let us first rewrite the symmetrised products as follows:
	\[
		\d\beta^{\dottimes n} = n!\d\beta^{\otimes n}
	\]
	thus
	\[
		\d\beta^{\dottimes n}\dottimes \dprod{\d\omega}{\d \omega}
		= \SymOp \d \beta^{\otimes n} \otimes \tprod{\d\omega}{\d\omega}
	\]
	
	Let us then compute the following
	\[\begin{aligned}
		D \lp \d \beta^{\otimes n} \otimes \tprod{\d\omega}{\d\omega} \rp
			&= -\frac{2}{\beta} \d \beta \otimes \d \beta^{\otimes n} \otimes \tprod{\d\omega}{\d\omega}
			- \frac{1}{\beta}\tprod{\d\omega}{\d \beta^{\otimes n} \otimes \lp \d \beta \otimes \d \omega + \d \omega \otimes \d \beta \rp}	
	\end{aligned}
	\]
	hence
	\[\begin{multlined}
		D \lp 
			\frac{1}{\beta^{n-1}}
			\d \beta^{\otimes n} \otimes \tprod{\d\omega}{\d\omega} \rp
		= \\
		-\frac{1}{\beta^n} \Big(
			(n+1) \d \beta \otimes \d \beta^{\otimes n} \otimes \tprod{\d\omega}{\d\omega}
			+ \tprod{\d\omega}{\d \beta^{\otimes n} \otimes \lp \d \beta \otimes \d \omega + \d \omega \otimes \d \beta \rp}
			\Big)
	\end{multlined}
	\]
	
	Now, we write $\SymOp_{\llbracket 2, n +3\rrbracket}$ for the symmetrising on the entries numbered from $2$ to $n+3$, so that
	\[\begin{aligned}
		\SymOp_{\llbracket 2, n +3\rrbracket} 
			D \lp 
				\frac{1}{\beta^{n-1}}
				\d \beta^{\otimes n} \otimes \tprod{\d\omega}{\d\omega} \rp
		&=
		D \lp \SymOp \frac{1}{\beta^{n-1}}
				\d \beta^{\otimes n} \otimes \tprod{\d\omega}{\d\omega} \rp\\
		&= D \lp  \frac{1}{\beta^{n-1}}
						\d \beta^{\dottimes n} \dottimes \dprod{\d\omega}{\d\omega} \rp	
	\end{aligned}
	\]
	We obtain
	\[\begin{aligned}
		\SymOp_{\llbracket 2, n +3\rrbracket} 
			D \lp 
				\frac{1}{\beta^{n-1}}
				\d \beta^{\otimes n} \otimes \tprod{\d\omega}{\d\omega} \rp
		&= -\frac{1}{\beta^n}\lp
			(n+1) \d \beta \otimes \lp \d \beta^{\dottimes n} \dottimes  \dprod{\d\omega}{\d\omega} \rp 
			+ 2\dprod{\d\omega}{\d \beta^{\dottimes n+1} \dottimes \d \omega}			
			\rp \\
		&= -\frac{1}{\beta^n}
			\d \beta^{\dottimes n+1} \dottimes  \dprod{\d\omega}{\d\omega}		
	\end{aligned}
\]
	which is enough to prove the recursion and conclude the proof.
\end{proof}

We can now give exact expressions for the higher covariant derivatives of $\zrot$, which are involved in the contribution of $\zrot$ to the Hessian sectional curvature of $\Gibbs$,
as functions of the expected values $\Ebo[\iota]$ and $\Ebo[\hi^k]$.
\begin{proposition}
	\begin{subequations}\label{eqnos:D1234zrot}
	\begin{align}
		D 	\zrot	&= \Ebo[\iota] \d(\beta\omega^2) = I \d(\beta\omega^2) \\
		D^2 \zrot	&= \Ebo[\hi^2] \d(\beta\omega^2)^{\otimes 2} + \Ebo[\iota] \beta \dprod{\d\omega}{\d \omega} \label{eqno:D2zrot}\\
		D^3 \zrot 	&= \Ebo[\hi^3]\d (\beta\omega^2)^{\otimes 3}
				  +	\Ebo[\hi^2] \beta \dprod{\d\omega}{\d \omega} \dottimes \d (\beta\omega^2)
				  - \Ebo[\iota] \d \beta\dottimes \dprod{\d\omega}{\d \omega}
		\\
		\begin{split}
		D^4	\zrot	&= \lp \Ebo[\hi^4] - 3\Ebo[\hi^2]^2 \rp \d (\beta\omega^2)^{\otimes 4}
			  +	\Ebo[\hi^3] \beta
			  	\tprod{\d\omega}{\d \omega}
			  	\dottimes (\d (\beta\omega^2) \otimes \d (\beta\omega^2))
			  \\ &
			  +	\Ebo[\hi^2] \lp
			  		-\d\beta \dottimes \dprod{\d\omega}{\d \omega} \dottimes \d(\beta\omega^2)
			  		+ \frac{\beta^2}2 \dprod{\d\omega}{\d \omega} \dottimes \dprod{\d\omega}{\d \omega}\rp
			  + \Ebo[\iota] \frac{\d\beta\dottimes \d\beta \dottimes \dprod{\d\omega}{\d \omega}}{\beta}
  		\end{split}
	\end{align}
	\end{subequations}
\end{proposition}
\begin{proof}
	This is a direct application of Formula~\eqref{eqnos:D1234z} to $\rho_*\nu_{\beta, \omega}$ using the Formula~\eqref{eqno:Dnbo} for the iterated covariant derivatives of $\beta\omega^2$.
	Recall that the symmetrised term 
	\[
				\frac{1}{j!}\frac{\SymOp}{n!} \lp
					\sum_{\substack{k : \llbracket 1, j\rrbracket \to\setN\\
							\sum k_i = n-j}}
						\binom{n}{k_1+1, \ \dots\ , k_j+1} 
				 D^{k} (\d(\beta\omega^2)^{\otimes j})
						\rp
	\]
	can be understood as the \enquote{naive} symmetrising of $D^{k}(\d(\beta\omega^2)^{\otimes j})$ (see the remark following Lemma~\ref{lmno:iXSymOp}).
	The form of the iterated covariant derivatives up to order 4 are given in Equations~\eqref{eqnos:D1234z}. Let us detail the computation of the symmetrisings:
	\begin{align*}
		\frac{\SymOp}{2} \lp D^2(\beta\omega^2) \otimes D\beta\omega^2 \rp
			&= \frac{\SymOp}{2} \lp \beta \dprod{\d \omega}{\d \omega} \otimes \d (\beta\omega^2) \rp
			= \beta \dprod{\d \omega}{\d \omega} \dottimes \d (\beta\omega^2)\\
		\frac{\SymOp}{8} \lp D^2(\beta\omega^2) \otimes D \beta\omega^2 \otimes D \beta\omega^2 \rp
			&= \frac{\SymOp}{8} \lp \beta \dprod{\d \omega}{\d \omega} \otimes \d (\beta\omega^2) \otimes \d (\beta\omega^2) \rp
			= \beta \tprod{\d \omega}{\d \omega} \dottimes (\d (\beta\omega^2) \otimes \d (\beta\omega^2))\\
		\frac{\SymOp}{6} \lp D^3(\beta\omega^2) \otimes D\beta\omega^2 \rp
			&= \frac{\SymOp}{6} \lp - \d \beta\dottimes \dprod{\d\omega}{\d\omega}\otimes \d (\beta\omega^2) \rp
			= - \d \beta\dottimes \dprod{\d\omega}{\d\omega}\dottimes \d(\beta\omega^2)\\
		\frac{\SymOp}{8} \lp D^2(\beta\omega^2) \otimes D^2\beta\omega^2 \rp 
			&= \frac{\SymOp}{8} \lp \beta \dprod{\d \omega}{\d \omega} \otimes \beta \dprod{\d \omega}{\d \omega} \rp
			= \frac12 \beta^2 \dprod{\d \omega}{\d \omega} \dottimes \dprod{\d \omega}{\d \omega}
	\end{align*}
\end{proof}

This is enough to obtain an expression of the Hessian curvature tensor, which would however require the somewhat cumbersome inversion of the metric.
As we expect, the asymptotics will be far simpler.
We will use the Landau notations $o$ and $O$ in the large $\beta\omega^2$ limit, for real-valued functions and sections of vector bundles with a metric.
\begin{definition}
	Let $E\to\Gibbs$ be a normed vector bundle, $\sigma$ a section of $E$ and $\alpha\in \setR$.
	We write $\sigma \eqbo o((\beta\omega^2)^\alpha)$ if $\frac{\lVert \sigma \rVert}{(\beta\omega^2)^\alpha} \tobo 0$.
	We write $\sigma \eqbo O((\beta\omega^2)^\alpha)$ if $\frac{\lVert \sigma \rVert}{(\beta\omega^2)^\alpha}$ is bounded as $\beta\omega^2 \to \infty$.
\end{definition}

\begin{lemma}\label{lmno:limEbohik}
	For any integer $k\geqslant 1$, 
	\[
		\Ebo[\hi^k] \eqbo O\lp \frac{1}{(\beta\omega^2)^k} \rp
	\]
	More precisely the function $\Ebo[\hi^k]\lp\beta\omega^2 \rp^k$ admits a finite limit as $\beta\omega^2\to \infty$.
\end{lemma}
\begin{proof}
	The integral $\Ebo[\iota]$ is invariant under rotation and purely depends on $\beta\omega^2$. 
	In this case, Theorem~\ref{thmno:GibbsLimit} implies that there is a well defined limit as $\beta\omega^2\to\infty$: in fact
	\[
		\Ebo[\iota] \xrightarrow[\beta\omega^2 \to \infty]{} \frac12 m R^2
	\]

	In particular, there is a smooth function $\hi_\infty = \iota - \frac12 m R^2$ such that $\lVert \hi - \hi_\infty \rVert_\infty \xrightarrow[\beta\omega^2 \to \infty]{} 0$.
	We decompose
	\[
		\hi^k
		= \lp \hi_\infty + \frac12 m R^2 - \Ebo[\iota] \rp^k
		= \sum_{p+q = k} \binom{k}{p}\hi_\infty^p \lp \frac12 m R^2 - \Ebo[\iota] \rp^q
	\]
	We first determine an equivalent for $\frac12 m R^2 - \Ebo[\iota]$:
	\[\begin{aligned}
		\int_{B(0,R)} \lp \frac12 m R^2 - \frac12 m \rho^2 \rp e^{\beta m \frac12 \omega^2 \rho^2} \frac{\d^3 q}{\Zrot}
			&= 
		\frac12 m\int_0^R \lp  R^2 - \rho^2 \rp
			e^{\beta m \frac12 \omega^2 \rho^2} 2\pi 2\sqrt{R^2-\rho^2} 
			\frac{\rho \d \rho}{\Zrot}\\
			&= 
		\frac12 m \int_0^{R^2} (R^2-x)^{\frac32}
			e^{\beta m \frac12 \omega^2 \rho^2} 
			\frac{2\pi \d x}{\Zrot}\\
			&\underset{\beta\omega^2 \to \infty}{\sim}
			\frac12 m
			\frac{\lp \frac12 \beta m \omega^2 \rp^{\frac32}}
				{\lp \frac12 \beta m \omega^2 \rp^{\frac52}}
			\frac{\Gamma\lp \frac52 \rp}{\Gamma\lp \frac32 \rp} 
			 = \frac32 \frac{1}{\beta\omega^2}
	\end{aligned}\]
	
	We now determine an equivalent for $\Ebo[\hi_\infty^p]$ with a similar computation:
	\[\begin{aligned}
		\int_{B(0,R)} \lp \frac12 m \rho^2 - \frac12 m R^2\rp^p e^{\beta m \frac12 \omega^2 \rho^2} \frac{\d^3 q}{\Zrot}
			&= 
		\lp -\frac12 m \rp^p
			\int_0^R \lp  R^2 - \rho^2 \rp^p
			e^{\beta m \frac12 \omega^2 \rho^2} 2\pi 2\sqrt{R^2-\rho^2} 
			\frac{\rho \d \rho}{\Zrot}\\
			&= 
		\lp -\frac12 m \rp^p
			\int_0^{R^2} (R^2-x)^{p+\frac12}
			e^{\beta m \frac12 \omega^2 \rho^2} 
			\frac{2\pi \d x}{\Zrot}\\
			&\underset{\beta\omega^2 \to \infty}{\sim}
		\lp -\frac12 m \rp^p
			\frac{\lp \frac12 \beta m \omega^2 \rp^{\frac32}}
				{\lp \frac12 \beta m \omega^2 \rp^{p+\frac32}}
			\frac{\Gamma\lp p +\frac32 \rp}{\Gamma\lp \frac32 \rp} 
			 =  \frac{\Gamma\lp p +\frac32 \rp}{\Gamma\lp \frac32 \rp} 
			 \lp \frac{-1}{\beta\omega^2}\rp^p
	\end{aligned}\]
	with $\frac{\Gamma\lp p +\frac32 \rp}{\Gamma\lp \frac32 \rp} = \prod\limits_{0 \leqslant j \leqslant p-1} (j + \frac32)$.
	By multiplication, we obtain the asymptotics of $\hi_\infty^p \lp \frac12 m R^2 - \Ebo[\iota] \rp^q$ up to $o\lp \frac{1}{(\beta\omega^2)^{p+q}} \rp$.
	The sum gives
	\[\begin{aligned}
		\Ebo[\hi^k]
			&= \sum_{p+q = k} \binom{k}{p}
				\Ebo[\hi_\infty^p]\lp \frac12 m R^2 - \Ebo[\iota] \rp^q\\
			&= \sum_{p+q = k} \binom{k}{p}
				\frac{\Gamma\lp p +\frac32 \rp}{\Gamma\lp \frac32 \rp}
				\lp \frac{-1}{\beta\omega^2} \rp^p
				\lp \frac32 \frac{1}{\beta\omega^2}\rp^q
				+ o\lp \frac{1}{(\beta\omega^2)^k} \rp\\
			&= \lp \sum_{p+q = k} \binom{k}{p}
				\frac{\Gamma\lp p +\frac32 \rp}{\Gamma\lp \frac32 \rp}
				\lp -1 \rp^p \lp \frac32 \rp^q
				\rp \frac{1}{(\beta\omega^2)^k}
				+ o\lp \frac{1}{(\beta\omega^2)^k} \rp\\
	\end{aligned}\]
	which proves the Lemma.
\end{proof}

\begin{remark}
	Note that for $k\geqslant 1$ the limit $\lim_{\beta\omega^2 \to\infty} \Ebo[\hi^k] (\beta\omega^2)^k$ is independent from $R$.
	This can be seen as a consequence of the scaling Equation~\eqref{eqno:muR1R2}: a rescaling of the sphere radius by a factor $\eta$ can be reflected as an extra factor of $\eta^2$ on the parameter $\beta$ in the measure, which contributes a factor $\eta^{-2k}$ to the asymptotic. At the same time, the function $\hi^k$ has a homogeneous behaviour of degree $2k$ under rescaling. As a consequence, the asymptotic at order $(\beta\omega^2)^{-k}$ is independent from $R$.
\end{remark}

Rather than the asymptotic behaviour of $\Ebo[\hi^k]$, we are interested in the asymptotic behaviour of the cumulants $\frac{\dd^k \zrot}{\dd (\beta\omega^2)^k} = \Ebo[f_k]$. They are strongly related to the moments $\Ebo[\hi^k]$ and we can express the ones in terms of the others using formal power series.

Let us work in the algebra of formal power series, which we equip with its \emph{(non-discrete) product topology}: a formal power series converges to a power series if and only if each of its coefficients converges to the coefficient of the same degree of the limit power series.
\begin{proposition}[Logarithm and exponential of a formal power series~\cite{BourbakiAlgebra47}]
	There is a logarithm application defined on formal power series with a constant coefficient of $1$:
	\[
		\log : 
		1 + \sum_{n\geqslant 1} a_n t^n
		\mapsto 
		\lp \sum_{n\geqslant 1} (-1)^{n-1}\frac{t^n}n \rp
			\circ
			\lp \sum_{n\geqslant 1} a_n t^n \rp
	\]
	The application $\log$ is continuous with respect to the product topology on the formal power series.
	It defines an inverse to the following application, defined on formal power series with vanishing coefficient of degree $0$:
	\[
		\exp : 
		\sum_{n\geqslant 1} a_n t^n
		\mapsto 
		\lp \sum_{n\geqslant 1} \frac{t^n}{n!} \rp
			\circ
			\lp \sum_{n\geqslant 1} a_n t^n \rp
	\]
\end{proposition}
The continuity is a direct consequence of the fact that the coefficients of the logarithm of a formal power series $F$ are polynomials in the coefficients of $F$. 
We now formulate the relation between the moments $\Ebo[\hi^k]$ and the cumulants $\Ebo[f_k]$ using power series:
\begin{proposition}\label{propno:Taylor}
	Let $(\beta, \omega)\in \Gibbs$. In the algebra of formal power series $\setR[[t]]$, the following equation holds:
	\[
		\sum_{n\geqslant 0} \frac{\Ebo[\hi^n]}{n!} t^n
		=
		\exp\lp
		\sum_{n\geqslant 2} \frac{\Ebo[f_n]}{n!} t^n		
		\rp
	\]
\end{proposition}
This relation is essentially the formal power series definition of the cumulants, applied to the centralised random variable $\hi = \iota - \Ebo[\iota]$. More detail can be found in~\cite{BrillingerTimeSeries, CombinatoricsCumulants, TensorStats} and general references on characteristic functions.

Since the logarithm of a power series is a continuous mapping, we can deduce the asymptotic behaviour of $\sum_{n\geqslant 2} \frac{\Ebo[f_n]}{n!} t^n$ from that of $\sum_{n\geqslant 0} \frac{\Ebo[\hi^n]}{n!} t^n$.

\begin{theorem}\label{thmno:fnequiv}
	For any integer $n\geqslant 2$, 
	\[
		\Ebo[f_n] \simbo (-1)^n \frac32 \frac{(n-1)!}{(\beta\omega^2)^n}
	\]
\end{theorem}
\begin{proof}
	According to Proposition~\ref{propno:Taylor}, for any $(\beta, \omega)\in \Gibbs$ the following relation between formal power series holds:
	\[
		\sum_{n\geqslant 2} \frac{\Ebo[f_n]}{n!}(\beta\omega^2)^n t^n		
		= \ln\lp 
		\sum_{n\geqslant 0} \frac{\Ebo[\hi^n]}{n!} (\beta\omega^2)^n t^n
		\rp
	\]
	We now compute the limit of $\sum_{n\geqslant 0} \frac{\Ebo[\hi^n]}{n!} (\beta\omega^2)^n t^n$ in the product topology:
	\[\begin{aligned}
	\lim_{\beta\omega^2\to\infty}
		\sum_{n\geqslant 0} \frac{\Ebo[\hi^n]}{n!} (\beta\omega^2)^n t^n
		&= \sum_{n\geqslant 0} \frac{1}{n!} \sum_{p+q = n} \binom{p+q}{p}
			\frac{\Gamma\lp p+\frac32 \rp}{\Gamma\lp \frac32 \rp} (-1)^p \lp \frac32 \rp^q t^{p+q}\\
		&= \sum_{p,q\geqslant 0} \frac{1}{p!q!}
			\frac{\Gamma\lp p+\frac32 \rp}{\Gamma\lp \frac32 \rp} (-t)^p \lp \frac{3t}2 \rp^q\\
		&= e^{\frac32 t} \sum_{p\geqslant 0}  \frac{\Gamma\lp p + \frac32\rp}{p! \Gamma\lp\frac32\rp} (-t)^p
	\end{aligned}\]
	Considering generalised binomial coefficients (cf. \cite{Feller}) we use the relation
	\[
		(-1)^p\binom{p+\frac12}{p} = \binom{-\frac32}{p}
	\]
	which allows recognising the series
	\[
		\sum_{p\geqslant 0}  \frac{\Gamma\lp p + \frac32\rp}{p! \Gamma\lp\frac32\rp} (-t)^p
		= \sum_{p\geqslant 0} \binom{-\frac32}{p} t^p
		= (1+t)^{-\frac32}
	\]
	We now obtain
	\[\begin{aligned}
	\sum_{n\geqslant 2} \frac{\Ebo[f_n]}{n!}(\beta\omega^2)^n t^n		
		\xrightarrow[\beta\omega^2 \to\infty]{}&
			\ln\lp 
			e^{\frac32 t} (1+t)^{-\frac32}
			\rp\\
		&= \frac32 (t-\ln (1+t))\\
		&= \frac32 \sum_{n\geqslant 2} \frac{(-t)^n}n
	\end{aligned}
	\]
	which gives the result.
\end{proof}
\begin{remark}
	It is possible to obtain the asymptotics for $\sum_{n\geqslant 0} \frac{\Ebo[\hi^n]}{n!} (\beta\omega^2)^n t^n$
	from the limit of $\frac{\Zrot(\beta\omega^2(1+t))}{\Zrot(\beta\omega^2)}e^{-\frac12\beta\omega^2R^2}$ as $\beta\omega^2\to \infty$, which can be calculated using the limits obtained in Section~\ref{secno:highvelocity}. This would however require more careful manipulations in order to relate the asymptotics of the functions with those of their formal Taylor series.
\end{remark}

%

We can now derive the asymptotics for the iterated covariant derivatives of $\zrot$ and therefore of $z$. 
Taking inspiration from Section~\ref{secno:rigidbodyspherical} we will use the coordinates $(u=2\beta^{-1/2}, \omega)$.
\begin{proposition}
	Let us define $I_{\infty} = \lim_{\beta\omega^2 \to \infty} I = m R^2$.
	We consider the coordinate system $(u=2\beta^{-1/2}, \omega)$.
	Then for the metric inverse to $g=D^2 z$ the following asymptotics hold:
	\begin{align*}
		\lVert \d u \rVert &\eqbo \beta^{-1/2} O(1)\\
		\lVert \d \omega \rVert_{\so_3\otimes T^*\Gamma} &\eqbo \beta^{-1/2} O(1)		
	\end{align*}
	and for any integer $n\geqslant 2$, 
	\[
		\beta^{-n/2}D^{n} z \eqbo 
			3 (n-1)! \d u^{\otimes n} 
		+ 	I_\infty \d u^{\dottimes (n-2)} \dottimes \tprod{\d\omega}{\d\omega}
		+ 	\beta^{-n/2}o(1)		
	\]

	In particular, for the iterated covariant differentials of $z$ of order $2$ to $4$ we obtain the following asymptotics:
		\begin{subequations}\label{eqnos:limD234z}
		\begin{align}
			\beta^{-1} D^2 z	&\eqbo
		  3 \d u \otimes \d u
		 + I_{\infty} \tprod{\d\omega}{\d \omega}
		 +\beta^{-1} o(1)\\
		\beta^{-3/2} D^3 z	&\eqbo
			6 \d u^{\otimes 3}
			+ I_{\infty} \d u\dottimes \tprod{\d\omega}{\d \omega}
			+ \beta^{-3/2} o(1) \\
		\beta^{-2}	D^4	z	&\eqbo
			18 \d u^{\otimes 4}
			+ I_{\infty} \d u\dottimes\d u \dottimes \tprod{\d\omega}{\d \omega}
			+\beta^{-2} o(1)
		\end{align}
		\end{subequations}
\end{proposition}

\begin{proof}
	Let us start with a few preliminary computations:
	\[
		\d u = 2\d \beta^{-1/2} = - \beta^{-3/2}\d \beta
	\]
	and
	\begin{equation}\label{eqno:dbolim}
	\begin{split}
		\d(\beta\omega^2)
			&= \omega^2 \d \beta + 2 \beta \sprod \omega {\d\omega}\\
			&= -\beta^{3/2}\omega^2 \d u + 2\beta \sprod \omega {\d\omega}\\
			&= \beta^{1/2}\lp - \beta\omega^2\d u + 2 \sprod {\beta^{1/2}\omega} {\d\omega}  \rp\\
			&\eqbo -\beta^{1/2}\lp \beta\omega^2 \d u + O\lp\lp \beta\omega^2\rp^{1/2} \rp \d\omega \rp
	\end{split}\end{equation}
	We will use Formula~\eqref{eqno:Dnzint} for $D^n \zint$:
	\begin{align*}
		D^n \zint &= \frac32 (n-1)! \beta^{n/2} \d u^{\otimes n}
	\end{align*}
	Finally, according to Theorem~\ref{thmno:fnequiv}, for $k\geqslant 2$ 
	\[
		\Ebo[f_k] \eqbo
		(-1)^k\frac32\frac{(n-1)!}{(\beta\omega^2)^k} + o\lp \frac{1}{(\beta\omega^2)^k} \rp
	\]
	We now use the Formula~\eqref{eqno:D2zrot} expressing $D^2\zrot$. We obtain immediately
	\[
		\beta^{-1} D^2 z
			\eqbo
		 \lp \frac32 + \frac32 + o(1)\rp \d u \otimes \d u
		 + \frac12 (I_{\infty}+o(1)) \dprod{\d \omega}{\d \omega}
		 + o(1) \d u\dottimes \d \omega
	\]
	which proves in particular that 
	\begin{subequations}\label{eqnos:limdudo}
	\begin{align}
		\beta \lVert \d u \rVert^2 &\tobo \frac{1}{3}\\
		\beta \lVert \d \omega \rVert^2_{\so_3\otimes T^*\Gamma} &\tobo \frac6{I_{\infty}}
	\end{align}
	\end{subequations}
	and we conclude that $\d u \otimes \d u$, $\d u\dottimes \d \omega$ and $\d \omega \otimes \d\omega$ are of class $\beta^{-1}O(1)$ as $\beta\omega^2 \to \infty$.
	Formulas~\eqref{eqnos:limdudo} along with Formula~\eqref{eqno:dbolim}, imply the following asymptotic:
	\[
		\d (\beta\omega^2)^{\otimes n} 
		\eqbo (-1)^n (\beta\omega^2)^n \beta^{n/2}\d u^{\otimes n}
			+ O \lp (\beta\omega^2)^{n-1/2} \rp 
	\]
	
	Let us now consider the expression of $D^{n} \zrot$ for $n\geqslant 1$ given by Theorem~\ref{thmno:DnzEVfj}:
	\[
		D^{n} \zrot	
		= \sum_{j=1}^{n}
			\frac{1}{j!}\EV_\xi[f_j]
			\frac{\SymOp}{n!} \lp
				\sum_{\substack{k : \llbracket 1, j\rrbracket \to\setN\\
					\sum k_i = n-j}}
					\binom{n}{k_1+1, \ \dots\ , k_j+1} 
	 D^{k} (\d\theta^{\otimes j})
				\rp
	\]
	Let us estimate the norm of $D^{l+1}(\beta\omega^2)$.
	According to Lemma~\ref{lmno:Dnbo}, we know that if $l\geqslant 1$, 
	\[
		D^{l+1}(\beta\omega^2)
		= \frac{(-1)^{l-1}}{\beta^{l-2}} \d \beta^{\dottimes l-2}\dottimes \dprod{\d\omega}{\d\omega}
		=
		\beta^{\frac{l+1}{2}} \d u^{\dottimes l-1}\dprod{\d\omega}{\d\omega}
		\eqbo O(1)
	\]
	For $j\in \llbracket 1, n-1 \rrbracket$ and $k$ a $j$-partition of $n-j$, $k$ contains at most $j-1$ zeroes, thus
	\[
						\sum_{\substack{k : \llbracket 1, j\rrbracket \to\setN\\
							\sum k_i = n-j}}
							\binom{n}{k_1+1, \ \dots\ , k_j+1} 
			 D^{k} (\d\theta^{\otimes j})
		\eqbo O(1) \d(\beta\omega^2)^{\otimes(j-1)}
		\eqbo O\lp (\beta\omega^2)^{j-1} \rp
	\]
	If furthermore $j>1$ then $\EV_\xi[f_j] = O \lp \frac{1}{(\beta\omega^2)^j} \rp$ thus
	\[
			\EV_\xi[f_j]
					\frac{\SymOp}{n!} \lp
						\sum_{\substack{k : \llbracket 1, j\rrbracket \to\setN\\
							\sum k_i = n-j}}
							\binom{n}{k_1+1, \ \dots\ , k_j+1} 
			 D^{k} (\d\theta^{\otimes j})
						\rp
			= O \lp \frac{1}{\beta\omega^2} \rp
	\]
	The term with $j=1$ is
	\[\begin{aligned}
		\EV_\xi[\iota] D^{n}(\beta\omega^2)
		&= \EV_\xi[\iota] (-1)^n \beta^{-(n-3)} \d \beta^{\dottimes (n-2)} \dottimes \dprod{\d\omega}{\d\omega}\\
		&= \beta^{n/2} \EV_\xi[\iota] \d u^{\dottimes (n-2)} \dottimes \dprod{\d\omega}{\d\omega}\\
		&\eqbo \beta^{n/2} \frac{I_\infty}2 \d u^{\dottimes (n-2)} \dottimes \dprod{\d\omega}{\d\omega}
		+ o(1)
	\end{aligned}
	\]
	The term with $j=n$ is
	\[\begin{aligned}
		\EV_\xi[f_n] \d (\beta\omega^2)^{\otimes n}
		&\eqbo
			(-1)^n \frac32 (n-1)!\lp \frac{1}{(\beta\omega^2)^n} + o\lp \frac{1}{(\beta\omega^2)^n} \rp\rp
			\lp (-1)^n (\beta\omega^2)^n \beta^{n/2}\d u^{\otimes n}
			+ o((\beta\omega^2)^n)\rp\\
		&\eqbo
			\frac32 (n-1)! \beta^{n/2} \d u^{\otimes n} + o(1)				
	\end{aligned}
			\]
	We obtain
	\[
	 D^n \zrot
	 \eqbo
	 \beta^{n/2}\lp  I_\infty \d u^{\dottimes (n-2)} \dottimes \tprod{\d\omega}{\d\omega}
	 + \frac32 (n-1)! \d u^{\otimes n} \rp
	 +  o(1)
	\]
	and finally
	\[
		\beta^{-n/2}D^n z \eqbo
		I_\infty \d u^{\dottimes (n-2)} \dottimes \tprod{\d\omega}{\d\omega}
			 + 3 (n-1)! \d u^{\otimes n} + \beta^{-n/2} o(1)			
	\]
	
%
%
\end{proof}

Note in particular the asymptotic contribution $\beta^{n/2} (-1)^n(n-1)!\d u^{\otimes n}$ to $D^n z$ due to $D^n \zrot$: it is identical to the exact contribution from $D^n \zint$. This fact can be interpreted as a manifestation of the principle of equirepartition of energy in the limit $\beta\omega^2\to \infty$.
Notably, the spherical confinement adds a contribution of $\frac32$ to the heat capacity of the gas at high velocity. Since $\Ebo[\hi^2] \simbo \frac32\frac1{(\beta\omega^2)^2}$, this factor $\frac32$ quantifies the standard deviation in the position of the particles at high velocity equilibrium.

We recognise in Formula~\eqref{eqnos:limD234z} the same form as the covariant differentials of the metric for the rigid body we derived in Section~\ref{secno:rigidbodyspherical}, with a heat capacity $C=3$ and a spherical inertia tensor of factor $I_\infty$.
In this sense, one can say that in the limit $\beta\omega^2\to\infty$, the rotating gas is \enquote{asymptotically rigid}.
Since we have done the algebraic calculations for the Hessian curvature in Section~\ref{secno:rigidbodyspherical}, we can simply use the result to obtain the limit of the Hessian curvature for the rotating perfect gas:
\begin{theorem}
	The Hessian curvature tensor $K$ of $\Gibbs$ has the following asymptotic behaviour:
	\begin{multline}
		\beta^{-2} K
		\eqbo
		30 \d u^{\otimes 4}
			- \frac12 \frac{I^2_{\infty}}3 \d \omega_i \otimes \d\omega_j\otimes \d\omega^i\otimes \d\omega^j\\
			+ \frac12 I_{\infty} \Big(
				\d u \otimes \d u \otimes \tprod{\d \omega}{\d \omega}
				+ \tprod{\d \omega\otimes \d u}{\d u \otimes \d \omega}\\
				+ \d u\otimes \tprod{\d\omega}{\d\omega} \otimes \d u
				+ \tprod{\d\omega}{\d\omega}\otimes \d u \otimes \d u
			\Big)
		+ \beta^{-2} o(1)
	\end{multline}
	with indices in the second term to indicate non-consecutive contractions (cf. Section~\ref{secno:notations})
\end{theorem}%
We directly obtain the Riemannian curvature tensor according to Proposition~\ref{propno:HessCurv}
\begin{theorem}
	The totally covariant Riemannian curvature tensor $\Riem$ of $g$ on $\Gibbs$ has the following behaviour:
	\begin{multline}
	\beta^{-2} \Riem
	\eqbo
		-\frac{1}{12}I^2_{\infty} \lp \d\omega_i \otimes \d \omega_j \otimes \d \omega^i \otimes \d\omega^j
			- \d\omega_j\otimes \d \omega_i \otimes \d\omega^i\otimes\d\omega^j \rp\\
		+ \frac{1}{4}I_{\infty}
			\big( \tprod{\d \omega\otimes \d u}{\d u \otimes \d \omega}
				- \d u \otimes \tprod{\d \omega}{\d u \otimes \d \omega}\\
				+ \d u\otimes \tprod{\d\omega}{\d\omega} \otimes \d u
				- \tprod{\d\omega\otimes\d u}{\d\omega} \otimes \d u
			\big)
		+ \beta^{-2} o(1)
	\end{multline}
	In particular,
	\[
		\Riem \eqbo -\frac{1}{12} \frac{g \owedge g}2 + o(1)
	\]
	with $\owedge$ denoting the Kulkarni-Nomizu product (more detail on the Kulkarni-Nomizu product can be found in~\cite{Besse, PetersenEdition3}).
\end{theorem}
We have proved that the asymptotic at dominant order of the Riemannian curvature tensor as $\beta\omega^2\to \infty$ corresponds to the constant sectional curvature Riemannian tensor of a hyperbolic space with sectional curvature $-\frac{1}{12}$.
Our result for the Hessian geometry is more precise: the asymptotic at dominant order of the Hessian curvature tensor corresponds to the Hessian curvature tensor of a rigid body with an inertia factor of $\frac12 m R^2$ and a heat capacity of $3$.

\printbibliography
\end{document}

%% file: def.tex
\newtheorem{theorem}{Theorem}[section]
\newtheorem{proposition}[theorem]{Proposition}
\theoremstyle{definition}
\newtheorem{definition}[theorem]{Definition}
\newtheorem{example}[theorem]{Example}

\newtheorem{lemma}[theorem]{Lemma}
\theoremstyle{remark}
\newtheorem*{remark}{Remark}

\newcommand{\setR}{\mathbb{R}}

\newcommand{\setN}{\mathbb{N}}
\newcommand{\setS}{\mathbb{S}}
\newcommand{\setP}{\mathbb{P}}

\newcommand{\argdot}{{\ \cdot \ }}

\DeclareMathOperator{\ad}{ad}
\DeclareMathOperator{\Ad}{Ad}

\DeclareMathOperator{\GL}{GL}

\DeclareMathOperator{\Diff}{Diff}

\newcommand{\Cinf}{\mathcal{C}^\infty}

\DeclareMathOperator{\Sym}{Sym}

\newcommand{\Lie}{\mathcal L}

\newcommand{\isom}{{\xrightarrow{\sim}}}

\DeclareMathOperator{\SO}{SO}

\renewcommand{\so}{\operatorname{\mathfrak{so}}}

\newcommand{\m}{\mathfrak m}

\newcommand{\Glie}{\mathfrak g}

\DeclareMathOperator{\Riem}{Riem}

\newcommand{\sprod}[2]{\left\langle #1 \,,\, #2 \right\rangle}

\renewcommand{\d}{{\operatorname d}} 
\newcommand{\dd}{\partial}
\newcommand{\der}[2]{\frac{\partial#2}{\partial#1}}

\newcommand{\U}{\mathcal U}

\newcommand{\lp}{\left(}
\newcommand{\rp}{\right)}

\newcommand{\suchthat}{\,/\,}

\newcommand{\lmtrois}[1]{\lambda^{\m(3)}}

